\documentclass[11pt]{article}

\usepackage{amsthm}
\usepackage[margin=1in]{geometry}

\usepackage{amsfonts}
\usepackage{amssymb}
\usepackage{amsmath}
\usepackage{hyperref}
\usepackage{mathrsfs}
\usepackage{centernot}
\usepackage{mathdots}
\usepackage{stmaryrd}
\usepackage[all]{xy}

\usepackage{mathtools}

\usepackage{color}

\newtheorem{thm}{Theorem}[section]
\newtheorem{lem}[thm]{Lemma}
\newtheorem{prp}[thm]{Proposition}

\newtheorem{dfn}[thm]{Definition}

\theoremstyle{definition}

\newtheorem{example}[thm]{Example}
\newtheorem{remark}[thm]{Remark}

\newtheorem{problem}[thm]{Problem}

\theoremstyle{plain}

\newcommand{\rem}[1]{}

\newcommand{\gap}[1]{{\color{blue} #1}}

\newcommand{\F}{\mathbb{F}}

\newcommand{\N}{\mathbb{N}}

\newcommand{\veps}{\varepsilon}
\newcommand{\vphi}{\varphi}

\newcommand{\suchthat}{\,:\,}
\newcommand{\where}{\,|\,}

\newcommand{\Circs}[1]{\left( #1 \right)}
\newcommand{\Squares}[1]{\left[ #1 \right]}

\newcommand{\ceil}[1]{\lceil {#1} \rceil}
\newcommand{\Ceil}[1]{\left\lceil {#1} \right\rceil}

\DeclareMathOperator{\Hom}{Hom} %
\DeclareMathOperator{\im}{im} %
\numberwithin{equation}{section} 

\usepackage{enumitem} 

\DeclareMathOperator{\dist}{\delta} 
\DeclareMathOperator{\maj}{maj} 

\newcommand{\Prob}{{\mathrm{Pr}}}

\author{Uriya A.\ First \\ University of Haifa \and Stav Lazarovici \\ University of Haifa}

\begin{document}

\title{Good Locally Testable Codes with Small Alphabet and Small Query Size}

\maketitle

\begin{abstract}
Ben-Sasson, Goldreich and Sudan 
\cite{BenSasson_2003_bounds_2query_codes} showed that a binary error correcting code  admitting a $2$-query tester cannot be good, i.e., it cannot have both linear distance and constant rate.
They also showed that there are no good codes if 
the alphabet is a finite field $\F$, the code is $\F$-linear, and the $2$-query tester is $\F$-linear. We show that those are essentially the only limitations on the existence of good locally testable codes (LTCs). That is,
there are good $2$-query LTCs on any alphabet with more than $2$ letters,
and  good $3$-query LTCs with a binary alphabet.
Similarly, there
are good $3$-query $\F$-linear LTCs, and
for every    $\F$-vector space $V$ of dimension greater than $1$,
there are good $2$-query LTCs with alphabet $V$ whose  tester is $\F$-linear.
This completely solves, for every $q\geq 2$ 
and alphabet (resp.\ $\F$-vector space) $\Sigma$, the question of whether there is a good $q$-query
LTC (resp.\ $\F$-LTC) with  alphabet $\Sigma$.
Our proof builds on      the recent   good   $2$-query $\F$-LTCs 
of the first author and Kaufman \cite{First_2024_cosyst_exp_posets_stoc}, 
by establishing a general method for reducing the alphabet size of a good low-query LTC.
\end{abstract}


\section{Introduction}

Throughout, $\Sigma$ is a (finite) alphabet and $\F$ is a finite field.
We let $C\subseteq \Sigma^n$ be an error correcting code,
which we think of as ranging in an infinite family of codes with block length tending
to infinity. The normalized Hamming distance in $\Sigma^n$ is denoted $\delta(\cdot,\cdot)$.
We recall other relevant coding theory terminology in Section~\ref{sec:prelim}.

Informally, the code $C\subseteq \Sigma^n$
is said to be a locally testable code (LTC)
if it
admits a randomized algorithm --- called a \emph{tester} --- that can estimate with high probability 
whether a word $w\in \Sigma^n$ is far from $C$ or not by reading only a constant number
of letters from $w$.
Formally, we require the tester for $C$ to read at most $q$ letters from $w$, 
accept all words in $C$,
and reject every $w\in \Sigma^n-C$ with probability at least $\mu\cdot \dist(w,C)$ for some $\mu>0$.
Here, $q$ and $\mu$ are  independent of $n$ and   $w$.
When we wish to specify the implicit constants $q$ and $\mu$,
we will    say that $C\subseteq \Sigma^n$
is a $q$-query LTC with soundness $\mu$.
Also, if not indicated otherwise, we   assume that any tester is \emph{non-adaptive}, i.e.,
it   decides which   positions  to read from the input word before   reading them.

The notion of an LTC arose in  the 1990s from 
the many works on the developing theories
of   \emph{property
testing} and \emph{probabilistically checkable proofs} (PCPs).
It first appeared in print in 
\cite[Dfn.~9]{Friedl_1995_total_degree_tests}
in a more lax version, in which the requirement
that $T$ rejects every $w\in \Sigma^n$ with probability
at least $\mu\cdot \dist(w,C)$ was imposed only for
words that are $\Omega(1)$-far from $C$. Such a codes 
are known as \emph{weak LTCs}.
The definition of LTCs which we use here, known as \emph{strong LTCs},  originates from \cite[Dfn.~2.2]{Goldreich_2006_LTCs_and_PCPs_of_almost_linear_length}, which
intitiated the systematic study of LTCs as  objects of independent interest.
See \cite{Goldreich_2011_locally_testable_codes_proofs} 
for a more detailed history of how LTCs
emerged.

Since their introduction, it was not clear whether there
are LTCs that are also \emph{good codes}, i.e., LTCs whose relative distance and rate
is bounded away from $0$; see \cite[Open Problem~2]{Goldreich_2011_locally_testable_codes_proofs}, for example. Indeed, Ben-Sasson, Goldreich and Sudan
\cite{BenSasson_2003_bounds_2query_codes} showed
that there are no good $2$-query LTCs on a binary alphabet,
and no $\F$-linear $2$-query LTCs.
(Here, an LTC is said to be $\F$-linear if its alphabet is
$\F$ and its tester checks $\F$-linear constraints.)
Further limitations on   $2$-query LTCs and
$3$-query LTCs appeared
in  \cite{BenSasson_2012_lower_bounds_on_weak_LTCs, Kol_2016_bounds_two_query_LTCs}. See also
\cite{BenSasson_2010_LTCs_red_Testers}.
Roughly at the same time, a series of works including
\cite{Polishchuk1994_nearly_lin_size_proofs,Dinur_2007_PCP_theorem_gap_amplification, BenSasson_2006_LTCs_prod_codes,
Viderman_2013_strong_LTCs,
Kopparty2017_high_rate_LTCs_subpoly_query, Gopi2018_LTCs_appr_GV_bound}
established the existence
of `almost' good LTCs (both weak and strong), e.g., 
LTCs with rate $\mathrm{poly}(\log n)^{-1}$ and constant relative distance,
or good codes admitting a tester reading $(\log n)^{O(\log \log n)}$ letters.
The existence of good LTCs was finally settled 
by Dinur--Evra--Livne--Lubotzky--Mozes \cite{Dinur_2022_ltcs_const_rate} and Panteleev--Kalachev \cite{Panteleev_2022_good_quantum_codes} (independently),
who showed that  good $\F$-linear LTCs (with large query size) do exist.
Shortly after,   the first author and Kaufman \cite{First_2024_cosyst_exp_posets_stoc}
(see also \cite{First_2024_cosyst_exp_posets_preprint}) showed
that there are also good $2$-query LTCs (with a large alphabet).
The latter codes are $\F$-LTCs, meaning that
their alphabet is an $\F$-vector space, and their tester checks   $\F$-linear constraints.

\subsection*{Main Results}

In this work, we completely resolve the question of what
are the integers
$q\geq 2$ and alphabets $\Sigma$ for which
there exists a good $q$-query LTC (resp.\ $\F$-LTC) with alphabet $\Sigma$.

\begin{thm}\label{TH:main_LTCs}
	Let $q\geq 2$ be an integer and let $\Sigma$ be an alphabet.
	\begin{enumerate}[label=(\roman*)]
		\item There exists a good $q$-query LTC with alphabet $\Sigma$
		if and only if $(q,|\Sigma|)\neq (2,2)$.
		\item Suppose $\Sigma$ is an $\F$-vector space.
		Then there exists a good $q$-query $\F$-LTC with alphabet
		$\Sigma$ if and only if $(q,\dim \Sigma)\neq (2,1)$.
	\end{enumerate}
\end{thm} 

Note that the  non-existence statements of Theorem~\ref{TH:main_LTCs} were
established by Ben-Sasson, Goldreich and Sudan \cite{BenSasson_2003_bounds_2query_codes}.
Theorem~\ref{TH:main_LTCs}
therefore says that the restrictions on the existence of good LTCs (resp.\ $\F$-LTCs) from 
\cite{BenSasson_2003_bounds_2query_codes}
are actually the only restrictions.

We prove the existence part of Theorem~\ref{TH:main_LTCs}
by introducing an  alphabet reduction method for LTCs and applying
it to the good $2$-query $\F$-LTCs of 
\cite[Thms.~9.7, 9.8]{First_2024_cosyst_exp_posets_preprint}
(see also \cite[Cor.~7.2]{First_2024_cosyst_exp_posets_stoc} where
the   construction is given for $\F=\F_2$).
This is 
the content of the following two theorems, which together
with \cite[Thms.~9.7, 9.8]{First_2024_cosyst_exp_posets_preprint} and \cite{BenSasson_2003_bounds_2query_codes}
imply Theorem~\ref{TH:main_LTCs}.

\begin{thm}\label{TH:main_non_lin}
	If there exists a good $2$-query LTC with a possibly adaptive tester
	on some alphabet $\Sigma$,
	then there exist
	\begin{enumerate}[label=(\roman*)]
		\item a good $2$-query LTC with alphabet $\Delta$ for every alphabet $\Delta$
	having more than $2$ letters, and
		\item a good $3$-query LTC with a binary alphabet.
	\end{enumerate}
\end{thm}

\begin{thm}\label{TH:main_lin}
	If there exists a good $2$-query $\F$-LTC, then there exist
	\begin{enumerate}[label=(\roman*)]
		\item a good $2$-query $\F$-LTC with alphabet $\Delta$ 
		for every $\F$-vector space $\Delta$ of dimension $>1$, and
		\item a good $3$-query $\F$-linear LTC.
	\end{enumerate}
\end{thm}

In both Theorem~\ref{TH:main_non_lin}
and Theorem~\ref{TH:main_lin}, the relative distance, rate and soundness
of the LTCs promised by the theorem
are at least proportional
to the corresponding parameters of the given LTC.
In Theorem~\ref{TH:main_lin},
all the parameters are diminished by a factor that is at most polynomial
in $|\Sigma|$, while in Theorem~\ref{TH:main_non_lin},
that factor is at most exponential in $|\Sigma|$; 
see Theorems~\ref{TH:linear_alphabet_reduction}
and~\ref{TH:alphabet_red}.
However, if the LTC given  in Theorem~\ref{TH:main_non_lin} is  
an $\F_2$-LTC, then the relative distance, rate and soundness
are diminished only by a factor polynomial in $|\Sigma|$;
see Theorem~\ref{TH:alphabet_red_semilin}. 
If the randomness complexity of the given LTC's tester
is $r(n)$ (a function of the block length $n$), then
the randomness complexity of \gap{}the promised LTC's tester is $\max\{r(n),\log_2 n\}+O(1)$;
see Remark~\ref{RM:randomness_comp}.

Analogous statements hold for $q$-query LTCs. Again, see Theorems~\ref{TH:linear_alphabet_reduction}
and~\ref{TH:alphabet_red}.

We hope that our alphabet reduction results will find other uses.

\begin{remark}
Some  methods for reducing the number of queries of a good LTC are also known.
Specifically, by \cite[Thm.~3.11, Prop.~3.6]{First_2024_cosyst_exp_posets_preprint}, 
if there exists a good $q$-query LTC  satisfying some assumptions
(which are satisfied by the good LTCs of \cite{Dinur_2022_ltcs_const_rate}) on alphabet $\Sigma$,
then there exists a good $2$-query LTC with alphabet $\Sigma'\subseteq \Sigma^q$. 
A different construction of this kind, which also needs
assumptions on the LTC, appears in \cite[Cor.~3.4]{Viderman_2013_strong_LTCs}.
It is further known that  
the existence of a good \emph{weak} $\F$-linear LTC implies the existence
of a good weak $3$-query $\F$-linear LTC \cite[Thm.~A.1]{BenSasson_2012_lower_bounds_on_weak_LTCs}.

We do not know if there is a query complexity reduction method that works
for a general good LTC.
\end{remark}

\subsection*{Proof Idea}

Our alphabet reduction method for LTCs (Theorems~\ref{TH:main_non_lin} and~\ref{TH:main_lin})
is simple in nature and consists
of   concatenating\footnote{
	Concatenation of codes is recalled in Section~\ref{sec:concat_codes}.
} the given $2$-query LTC $C\subseteq \Sigma^n$ with a special code 
$D\subseteq \Delta^k$
(which remains fixed as $n$ grows)
and showing that the concatenated code
$C\circ D\subseteq \Delta^{nk}$ is also a good $2$-query LTC.
While concatenation is a well-known method for reducing a good code's
alphabet size, it usually fails to preserve the property of being a $2$-query LTC.
The novelty of our approach is in showing that  the inner code $D\subseteq \Delta^k$ (and the identification of $\Sigma$ with $D$)
can be chosen so that the property of being a $2$-query LTC will persist.

Let us now explain in more detail how our concatenation approach
works in the context of Theorem~\ref{TH:main_non_lin}(i).
To begin, 
let us first see why $C\circ D\subseteq \Delta^{nk}$ is in general not a $2$-query
LTC. 
Let $T$ denote the $2$-query tester of $C$,
and choose  some encoding function $f:\Sigma\to D\subseteq \Delta^k$. Then each codeword
of $C\circ D$ 
has the form
\[
f(w_1)f(w_2)\cdots f(w_n)\in \Delta^{nk}
\]
for some $w\in C$. 
Let   $u\in \Delta^{nk}$ and write $u=u_1u_2\cdots u_n$ with $u_1,\dots,u_n\in \Delta^k$.
A natural way to test whether $u$ is in $C\circ D$ would be to try one of the following:
\begin{enumerate}[label=(\arabic*)]
	\item Check that $u_i\in D$ for some $i\in [n]$;
	\item Emulate the tester $T$ of $C$: run   $T$  to get two coordinates $i,j\in[n]$
	to read,
	find the letters $a,b\in \Sigma$ satisfying $f(a)=u_i$, $f(b)= u_j$
	(or reject if there are no such $a,b$), and return
	what $T$ would have returned upon reading $a$ and $b$
	in positions $i$ and $j$.
\end{enumerate}
A naive implementation of (1) and (2) requires making  $k$ or $2k$ 
queries to $u$, so the  $2$-query testability is seemingly lost. 

We   solve the large number of queries in (1) by
requiring that $D\subseteq \Delta^k$ is defined by $2$-letter constraints.
Then, instead of checking whether $u_i \in D$, we  
check whether $u_i$ satisfies a random $2$-letter constraint.
(Note that since $D\subseteq \Delta^k$ remains fixed as $n$
grows,   we do not need the stronger assumption that $D\subseteq \Delta^k$
is included in a family of codes that is a $2$-query LTC.)

Reducing the number of queries in (2) is  possible under special assumptions
on $T$ and $f$. 
Given a  random seed $s$ for  $T$,
denote by $i_{s}, j_{s}\in [n]$ the coordinates
queried by $T$, and by $T^s(w)$ the output of $T$
on $w\in\Sigma^n$ when given the seed $s$.
What we need from $T$ and $f$ is that for
every $s$,  there are coordinates $i',j'\in [k]$ (depending on $s$)
such that one can determine $T^s(w)$ --- which
depends only on $w_{i_s}$ and $w_{j_{s}}$ ---
from the $i'$-th letter of $f(w_{i_s})$
and the $j'$-th letter of $f(w_{j_{s}})$.
In other words, writing $f_{i'}:\Sigma\to \Delta$
for the $i'$-th coordinate of $f$ (so that $f(a)=(f_1(a),\dots,f_k(a ))$),
we want to be able to determine $T^s(w)$
from $f_{i'}(w_{i_s})$ and $f_{j'}(w_{j_s})$ for some $i',j'$ depending on $s$.
When this holds for every seed $s$, we   say that $T$ is \emph{$f$-compatible};
see Section~\ref{sec:concat_testers} for details.
Provided that $T$ is $f$-compatible, we can emulate the action of $T^s$
on $u=u_1\cdots u_n $ by reading just two letters from $u$,
namely, $(u_{i_s})_{i'}$ and $(u_{j_{s}})_{j'}$, and thus perform (2) using 
only $2$  queries.

We show in Theorem~\ref{TH:comp_of_codes_w_testers} that
if $D\subseteq \Delta^k$ is defined by $2$-letter constraints
and $T$ is $f$-compatible, then the concatenated code
$C\circ D$ is indeed a $2$-query LTC. 
With this at hand, we need to find   a code $D\subseteq \Delta^k$
and an encoding function $f:\Sigma \to D $
such that
(i) $T$ is $f$-compatible, and
(ii) $D\subseteq \Sigma^k$ is defined by
$2$-letter constraints.

The code $D\subseteq \Delta^k$ that we choose for the task
is a generalization of the \emph{long code} of \cite[\S3]{Bellare_1998_free_bits_PCPs}.
Formally, letting $f_1,\dots,f_k$ denote \emph{all} the functions from $\Sigma$
to $\Delta$ (so that $k=|\Delta|^{|\Sigma|}$), we choose $D$ to be the code $L(\Sigma,\Delta):=\{(f_1(a),\dots,f_k(a))\where a\in \Sigma\}\subseteq \Delta^k$
and take $f:\Sigma\to D$ to be the obvious encoding
$f(a)=(f_1(a),\dots,f_k(a))$. We call $L(\Sigma,\Delta)$ a \emph{generalized long code}; 
the
usual long code is the special case where $\Sigma=\{0,1\}^t$ and $\Delta=\{0,1\}$. 
The rationale behind this choice of $D$ is that having any
function $g:\Sigma\to \Delta$
in the collection  $\{f_i\}_{i=1,\dots,k}$
helps in securing
that $T$ is $f$-compatible. (This is still not enough to guarantee
$f$-compatibility, and we solve this below.)
The drawback of choosing $D$ to be  $L(\Sigma,\Delta)$ is that it
{\it a priori} it does not
fulfill the requirement   of being defined by $2$-letter constraints.
We prove this nontrivial statement 
in Theorem~\ref{TH:gen_long_code_2_test} under the assumption $|\Delta|\geq 3$
(the hardest case is $|\Delta|=3$).
By contrast, the `usual' long code $L(\{0,1\}^t,\{0,1\})$ is \emph{not} defined
by $2$-letter constraints; see Remark~\ref{RM:long_code_not_2_test}.

In order to finish, 
it   remains to check that  the tester $T$ is $f$-compatible.
We show in Proposition~\ref{PR:sep_nec_and_suff} that the $f$-compatibility
of $T$ is equivalent to another condition called \emph{$\Delta$-separability}.
While this condition fails in general, 
 we show in 
Theorem~\ref{TH:reduction_to_sep} that $T$ can be replaced with
a $\Delta$-separable $2$-query tester  with soundness that is proportional to
that of $T$. Conveniently, this replacement is also the usual trick
to turn an adaptive tester into a nonadaptive tester, so we can also treat
non-adaptive $2$-query testers.
After that replacement, we can finally conclude that the concatenated code $C\circ D\subseteq \Delta^{nk}$ is a $2$-query LTC.

The proof of Theorem~\ref{TH:main_lin}(i) follows
a similar path except that
we define $D\subseteq \Delta^k$ by choosing $f_1,\dots,f_k$
to be all the $\F$-linear functions from $\Sigma$ to $\Delta$.
This gives rise to what we call a \emph{generalized Hadamard code};
see \S\ref{subsec:gen_Hadamard}. 
Similarly to the generalized
long code, it can be defined by $2$-letter constraints
if and only if $\dim \Delta \geq 2$ (Theorem~\ref{TH:gen_Hadamard_2_test}).

\begin{remark}\label{RM:soundness_loss}
Passing from $C\subseteq \Sigma^n$ to $C\circ D\subseteq \Delta^{nk}$ as described
above
scales down the soundness by a constant factor depending on $\Sigma$ and $D$.
There are three sources for this scale-down. 

The first   is the soundness of the $2$-query tester for $D$ (cf.\ (1) above) --- we use
the naive lower bound $\frac{1}{R}$ where $R$ is the number of $2$-letter constraints
defining $D$;
for the generalized long code this factor is $|\Delta|^{-O(|\Sigma|)}$.
This can probably be improved significantly. That is, when $|\Delta|\geq 3$, we expect that 
the generalized long code $L(\Sigma,\Delta)\subseteq \Delta^{|\Delta|^{|\Sigma|}}$
should have a $2$-query tester with soundness that is \emph{independent of $|\Sigma|$}.
Similarly, in the linear case, we expect the generalized Hadamard code (\S\ref{subsec:gen_Hadamard})
to have the same property when $\dim \Delta\geq 2$. We pose this as a Problem~\ref{PB:long_code_LTC} below.

The  second source for loosing on the soundness is the block length $k$ of $D\subseteq \Delta^k$.
This $\frac{1}{k}$ factor comes up in the analysis of the tester of $C\circ D$,
and is forced in situations in which the coordinates $i',j'\in [k]$
from the definition of $f$-compatibility (which are determined by   the random seed of $T$) are very unevenly distributed.
See Remark~\ref{RM:block_len_loss} for details and possible ways one might overcome it (which 
were not pursued in this work).

The final cause for loosing on the soundness is the replacement of $T$
by a $\Delta$-separable tester. In more detail, in the non-linear case, 
the trick is to first guess what are the two letters 
that will be read by $T$, and   apply $T$ only if these letter are actually
encountered (otherwise accept the word). 
This multiplies the soundness by a factor of $|\Sigma|^{-2}$. We do not know how
to avoid this factor in general. However, in the linear case, $T$ is morally
close to being $\Delta$-separable (it is always $\Delta^m$-separable for $m=\ceil{\frac{2\dim \Sigma}{\dim \Delta}}$), and one can use a different trick 
which reduces the soundness by a factor
of $O(\log |\Sigma|)$.
\end{remark}

\begin{problem}\label{PB:long_code_LTC}
	(i) Let $\Delta$ be an alphabet of $3$ or more letters. Is the family
	of generalized long codes $\{L(\{1,\dots,k\},\Delta)\subseteq \Delta^{|\Delta|^k}\}_{k\geq 1}$
	a $2$-query LTC?
	
	(ii) Let $\Delta$ be an $\F$-vector space of dimension $2$ or more.
	Is the family of generalized Hadamard codes $\{\Hom_{\F}(\F^k,\Delta)\subseteq \Delta^{|\Delta|^{|\F^k|}}\}_{k\geq 1}$ (see \S\ref{subsec:gen_Hadamard}) a $2$-query LTC?
\end{problem}

\subsection*{Relation to Other Works on Alphabet Reduction}

	Concatenation with the `ordinary' long code
	and Hadamard code was used in the 
	literature
	for alphabet
	reduction in the construction of PCPs and $2$-query LTCs.
	For example, in \cite[\S\S7--8]{Dinur_2007_PCP_theorem_gap_amplification},
	resp.\ \cite[\S6.4]{Meir_2009_comb_const_of_LTCs},      this technique
	is used to construct   non-linear, resp.\ linear, weak $2$-query LTCs  having linear distance and inverse-poly\-logarith\-mic rate. The same method was pushed further
	in \cite{Viderman_2013_strong_LTCs} to construct (strong) $\F_2$-linear $3$-query LTCs (with a binary
	alphabet) and $\F_2$-LTCs with an $8$-letter alphabet  having    similar distance and rate.

	In a nutshell, the idea of these alphabet reductions is 
	to encode both  the  {letters} 
	and  the   {constraints} of the code; this is different
	from our approach which encodes only the  {letters}. 
	To explain this, let us describe an oversimplified version
	of the alphabet reduction of \cite{Dinur_2007_PCP_theorem_gap_amplification}:
	We start with a code $C\subseteq \Sigma^n$ having a $2$-query tester $T$.
	Choose some embedding of $\Sigma$ in $\{0,1\}^t$
	and use it to view words in $\Sigma^n$ as strings of $tn$ bits.
	Think of $T$ as checking at random one of (say) $R$
	$2$-letter constraints which define $C\subseteq \Sigma^n$.
	Since $\Sigma\subseteq\{0,1\}^t$,
	every such constraint is a circuit on $2t$ bits.
	Now construct a new code $C'\subseteq \{0,1\}^{N}$ by modifying
	the codewords $w\in C$ as follows. First, view $w$ as a $tn$-bit string using that $\Sigma\subseteq \{0,1\}^t$.
	Then, for each of the $R$ constraints defining $C$, attach to $w$ the long-code 
	encoding
	of the $2t$ bits participating in that constraint.\footnote{
		This is a little different from \cite[\S7]{Dinur_2007_PCP_theorem_gap_amplification},
		where a certain puncturing of the long code is used. 
	} 
	This results in a string of $N=tn+2^{2^{2t}}R$ bits, 
	and we take $C'$
	to be the collection of all such strings. 
	Since the long-code encoding of $2t$ bits
	includes the evaluation of any circuit   on those bits,
	if $w'\in \{0,1\}^N$ is obtained from some $w\in \Sigma^n$ by the above procedure, 
	then we can emulate
	the action of $T$ on $w$ by reading just $1$ bit from $w'$.
	The long code\footnote{
		More precisely, its puncturing that is used in \cite[\S7]{Dinur_2007_PCP_theorem_gap_amplification}.
	} has two additional relevant properties.
	First, it has a $3$-query tester. Second, it has a $2$-local decodability
	property --- if $y\in\{0,1\}^{2^{2^{2t}}}$ is close to the encoding of some $x\in \{0,1\}^{2t}$, then
	there is a randomized
	algorithm which can determine $x_i$ with high probability by reading just $2$ letters from $y$.%
	\footnote{
		The algorithm: choose uniformly at random
		two functions $f,g:\{0,1\}^{2t}\to \{0,1\}$ such that $f +g$
		is the $i$-th coordinate function $e_i:\{0,1\}^{2t}\to\{0,1\}$,
		and return $y_f+y_g$. See \cite[Thm.~7.1]{Dinur_2007_PCP_theorem_gap_amplification} for
		a precise statement. See \cite[Lem.~6.21]{Meir_2009_comb_const_of_LTCs} for a corresponding
		statement for the Hadamard code.}
	Using these properties and other mild assumptions, 
	it can be shown that $C'\subseteq \{0,1\}^N$
	has a $3$-query tester   with soundness that is proportional to that of $T$.
	Briefly, the idea is that given $u\in\{0,1\}^N$, 
	we can detect if $u$ is very corrupted in its last $2^{2^{2t}}R$ bits (the long-code encoding
	of the constraints)
	using the   $3$-query tester
	of the long code. If that is not the case, then we can use the $2$-local decodability property
	of the long code to check if the last $2^{2^{2t}}R$ bits of $u$
	are consistent with the first $tn$ bits (this requires querying $2+1=3$ bits from $u$). When they are almost consistent, meaning that
	$u$ is very close to  a word coming from some $w\in \Sigma^n$, we can emulate $T$ on $w$.
	
	There are two reasons why we cannot use this alphabet reduction approach
	to prove our main results. First, the above construction
	increases
	the block length by $O(R)$, where $\log_2 R$ is the randomness complexity
	of $T$, while leaving the message length the same. Therefore, applying
	it to a good LTC would 
	result in a good LTC only if $R=O(n)$, equiv.\ the randomness complexity of $T$ is $\log n+O(1)$.
	By contrast, our alphabet reduction method 
	(Theorems~\ref{TH:main_non_lin} and~\ref{TH:main_lin})	
	makes no assumptions on the randomness complexity.	
	Second,   the alphabet reduction just described (loosely) relies on the fact
	that the long code (resp.\ Hadamard code)  is $2$-locally decodable, 
	which is one reason why the resulting tester
	needs to query $2+1=3$ letters (the other reason is that the long code's tester queries $3$ letters). 
	If we were to use this approach to get a $2$-query LTC,
	we would need to replace the long code with a   $1$-locally decodable code,
	and there are no such useful codes.
	That said, it seems likely that applying the alphabet reduction
	of  \cite{Dinur_2007_PCP_theorem_gap_amplification}, \cite{Meir_2009_comb_const_of_LTCs}, \cite{Viderman_2013_strong_LTCs}
	to a good $2$-query LTC (or even any good LTC)
	whose tester has randomness complexity $\log_2(n)+O(1)$,
	e.g., those of
	\cite{First_2024_cosyst_exp_posets_stoc}
	(or \cite{Dinur_2022_ltcs_const_rate}, \cite{Panteleev_2022_good_quantum_codes}),
	would result in a good $3$-query LTC on a binary alphabet.
	
	Finally, in terms of soundness, the alphabet
	reduction of   \cite{Dinur_2007_PCP_theorem_gap_amplification}, etc.\
	is more economical than ours as it decreases the soundness by a 
	constant factor (independent of $|\Sigma|$); see Remark~\ref{RM:soundness_loss} for what causes the soundness loss in our alphabet reduction.
	As for rate and distance, our alphabet reduction reduces the rate
	by a smaller factor than that of {\it op.\ cit.},
	and both alphabet reductions reduce the distance by a constant factor independent of $|\Sigma|$.

\subsection*{Organization}

The paper is organized as follows:
Section~\ref{sec:prelim} is preliminary and recalls
necessary definitions and facts about   codes and testers.
In Section~\ref{sec:concat_codes}, we recall
concatenation of codes.
In Section~\ref{sec:concat_testers},
we prove a theorem giving sufficient conditions for the concatentation
of two codes to be a $q$-query LTC.
The inner codes to which this result will
be applied --- the generalized long code and the generalized
Hadamard code --- are presented in 
Section~\ref{sec:gen}. In that section, it is also shown that
these codes are defined by $2$-letter constraints.
Section~\ref{sec:sep} is concerned with showing that one can replace
the tester of any LTC by another tester to which
our result about concatenation of LTCs can be applied.
Finally, in Section~\ref{sec:alpha_red}, all previous results are combined
to prove Theorems~\ref{TH:main_non_lin} and~\ref{TH:main_lin}.

\subsection*{Acknowledgements.}
We are grateful to the anonymous referees for many useful suggestions,
and to Esty Kelman for pointing out a mistake in the proof of Theorem 4.2 in an earlier
version.
This
research was supported by ISF grant no.\ 721/2024.

\section{Preliminaries}
\label{sec:prelim}

\subsection{General Conventions}

Throughout this paper, $\F$ is a finite field, and $[n]$ denotes the set $\{1,\dots,n\}$. 
Vector spaces are over $\F$ and are assumed to be finite-dimensional.
(Unlike some other texts, the letter $q$ will be reserved for the number of queries performed by a tester
and  will \emph{not} denote  the size of $\F$.)

An alphabet is a finite set with at least two elements.

Recall that a distribution on a countable set $X$
is a collection of non-negative real numbers $p=(p_x)_{x\in X}$
adding up to $1$; the number $p_x$ is the probability of drawing
$x$ when sampling an element from $X$ according to $p$. 
We write   $x\sim p$ to indicate that $x\in X$
is chosen at random accroding to the distribution $p$.
For a finite set $X$, we write $x\sim X$ to denote
that $x$ is chosen from $X$ uniformly at random.

\subsection{Error Correcting Codes}

Let $\Sigma$ be an alphabet and $n\in\N$.
We write $\Sigma^n$ for the set of $n$-letter
words in the alphabet $\Sigma$, and unless indicated otherwise,
write $w_i$ for the $i$-th letter of a word $w\in \Sigma^n$.
As usual, the normalized Hamming distance   on $\Sigma^n$, denoted $\dist(\cdot,\cdot)$,
is given by $\dist(u,v)=\frac{1}{n}\cdot \#\{i\in\{1,\dots,n\}\suchthat u_i\neq v_i\}$ for all $u,v\in\Sigma^n$.

In this work, an \emph{error correcting code},
or a  \emph{code} for short, with alphabet $\Sigma$ and block length $n$ is a nonempty
subset $C\subseteq \Sigma^n$. 
Recall that the   \emph{relative distance} of   $C\subseteq \Sigma^n$ is\label{symdef:delta} 
\[\delta(C):=
\min\{\dist(u,v)\where u,v\in C,\, u\neq v\}\]
and its \emph{rate} is\label{symdef:rate}
\[
r(C) := \frac{\log_{|\Sigma|}|C|}{n}.
\]
We say that $C$ has relative distance $\delta$ (resp.\ rate $r$)
when $\delta(C)\geq \delta$ (resp.\ $r(C)\geq r$)
and add the word ``exactly'' to indicate that equality holds.

It is common to think of a code $C\subseteq \Sigma^n$ as ranging in  a \emph{family of codes}, i.e.,
a sequence of codes $\{ C_i\subseteq \Sigma^{n_i}\}_{i\in \N}$ with
$n_i$ tending to $\infty$. 
Recall that the family  $\{ C_i\subseteq \Sigma^{n_i}\}_{i\in \N}$ 
is said to be a \emph{good code}  if there are $r,\delta>0$ such that
$r(C_i)\geq  r$ and     $\delta(C_i)\geq \delta$ for all $i$.
In this case, we also say that the family $\{ C_i\subseteq \Sigma^{n_i}\}_{i\in \N}$
has relative distance $\delta$ and rate $r$.

\subsection{Testers}
\label{subsec:testers}

Let $q\in\N$. 
Recall that
a \emph{$q$-query tester}, or a \emph{$q$-tester} for short,
for a code $C\subseteq \Sigma^n$ is a randomized algorithm $T$ which, given oracle access to some 
$w\in\Sigma^n$, reads at most $q$ letters from $w$ and returns $1$ --- meaning `accept' --- or $0$ --- meaning `reject' ---
subject to the requirement that any $w\in C$ is accepted.
The tester $T$ is said to have soundness $\mu$ ($\mu\geq 0$) if 
\[
\Prob(T(w)=0)\geq \mu\dist(w,C)
\qquad\forall w\in \Sigma^n.
\]
The tester $T$ is called  \emph{non-adaptive}
if it does not   use information from previous queries to $w$
to determine which position to read next.
\emph{Unless explicitly indicated, we always 
assume that  testers
are  non-adaptive.}

Observe that in order to specify a (non-adaptive)  $q$-tester
for $C\subseteq \Sigma^n$, it is enough to give   the following data:
a finite set $I$, a probability distribution $p=(p_i)_{i\in I}$ on $I$,
a $q$-tuple  $(a_1^{(i)},\dots,a_q^{(i)})\in [n]^q$ for each $i\in I$,
and a function $T^{(i)}:\Sigma^q\to \{0,1\}$ for every $i\in I$.
The corresponding tester $T$ then works as follows:
Given $w\in \Sigma^n$, choose $i\in I$ according to the distribution 
$p$, read the letters in positions
$a_1^{(i)},\dots,a_q^{(i)}$ from $w$ and return
$T^{(i)}(w_{a_1^{(i)}},\dots,w_{a_q^{(i)}})$.
We write this simply
as
\[
T=(T^{(i)},(a_1^{(i)},\dots,a_q^{(i)}),p_i)_{i\in I}.
\]

A code $C\subseteq\Sigma^n$ is said to be \emph{defined by $q$-letter
constraints} if it has a $q$-query tester with positive soundness.
This is equivalent to saying that there is a list
of constraints on words in $\Sigma^n$,
each   involving $q$ or less letters, such that
$C$ is the set of words in $\Sigma^n$ satisfying all those constraints.

A family of codes $\{C_i\subseteq\Sigma^{n_i}\}_{i\in\N}$ is
called a \emph{$q$-query locally testable code} ($q$-query LTC)
if there is $\mu>0$ such that each code in the family admits a $q$-tester
with soundness $\mu$. In this case, we also say that the $q$-query LTC $\{C_i\subseteq\Sigma^{n_i}\}_{i\in\N}$
has soundness $\mu$.
A \emph{locally testable code} (LTC) is a family of codes that is a $q$-query LTC
for some $q\in\N$.

\subsection{Linear Codes}
\label{subsec:lin_codes}

Let $\F$ be a finite field,
and suppose that the alphabet $\Sigma$ is also an $\F$-vector space
(hence $\Sigma\neq 0$).
A code $C\subseteq \Sigma^n$ is said to be an \emph{$\F$-code}
if $C$ is a subspace of $\Sigma^n$.
An \emph{$\F$-linear code} is an $\F$-code with alphabet $\F$.

Let $C$ be an $\F$-code.
We say that a $q$-tester $T=(T^{(i)},(a_1^{(i)},\dots,a_q^{(i)}),p_i)_{i\in I}$
is \emph{$\F$-linear} if
for each $i\in I$, there is a linear subspace $V_i\subseteq \Sigma^q$
such that
\[
T^{(i)}(u)=\left\{
\begin{array}{ll}
1 & u\in V_i \\
0 & u\notin V_i
\end{array}
\right.
\]
for all $u\in \Sigma^q$. For example, when $\Sigma=\F$ and  $V_i$ has codimension $1$ in $\Sigma^q$,
the test $T^{(i)}$ merely checks that the input $w\in\F^n$ 
satisfies a single homogeneous linear equation involving at most $q$ of the coordinates.

Any $q$-tester for an $\F$-code can be replaced
by an $\F$-linear $q$-query tester having soundness greater than or equal to the soundness of the original tester.

A $q$-query \emph{$\F$-LTC} (resp.\ \emph{$\F$-linear LTC}) is a family of $\F$-codes (resp.\ $\F$-linear codes)
for which there is $\mu>0$ such that each code in the family
admits an $\F$-linear $q$-tester with soundness $\mu$.

\section{Concatenation of Codes}
\label{sec:concat_codes}

We proceed by recalling concatenation of codes, setting notation along the way.
Let $C\subseteq \Sigma^n$ be a code with alphabet $\Sigma$,
let $\Delta$ be another alphabet, and let $f:\Sigma\to \Delta^k$ be an injective function
with image $D$. We think of $D$ as a code inside $\Delta^k$ and of $f:\Sigma\to D\subseteq \Delta^k$
as its encoding function.
Recall that the \emph{concatenated code} $C\circ D\subseteq \Delta^{kn}$ is obtained by
encoding the letters of each codeword in $C$ using the code $D$.
Formally,
\[
C\circ D:=\{f(w_1)\cdots f(w_n)\where w\in C\}\subseteq \Delta^{kn}.
\]
The codes  $C$ and $D$
are often called the \emph{outer code} and \emph{inner code},
respectively.
When the encoding map $f$ is not clear from the context, we shall
write $C\circ_f D$ for $C\circ D$.  For more details, see \cite[Chp.~12]{Roth_2006_coding_thy}
or \cite[\S10.1]{Guruswami_2025_ess_coding_thy_preprint}, for instance.

Note that if $C$ is an $\F$-code, $\Delta$ is an $\F$-vector space
and $f:\Sigma\to \Delta^k$ is $\F$-linear, then $C\circ D$
is an $\F$-code. When $\Delta=\F$, the code $C\circ D$
is moreover $\F$-linear.

The behavior of the relative distance and rate of $C\circ D$
is well-known and summarized in the following proposition.

\begin{prp}\label{PR:concat_rate_dist}
With notation as above, we have
\[
\delta(C\circ D)\geq \delta(C)\delta(D)
\qquad\text{and}\qquad 
r(C\circ D)=r(C)r(D).
\]
\end{prp}

In  applications of concatenation in this paper, the outer code $C$ will  
range over a family of codes $\{C_i\subseteq \Sigma^{n_i}\}_{i\in\N}$, while the inner code $D\subseteq \Delta^k$ will  remain  fixed.
Then, by Proposition~\ref{PR:concat_rate_dist}, the family
$\{C_i\circ_f D\subseteq \Sigma^{kn_i}\}_{i\in \N}$
will be a good code as long as  $\{C_i\subseteq \Sigma^{n_i}\}_{i\in\N}$
is a good code; this holds no matter
how poor the relative distance and rate of $D$
are.

\section{Concatenation of Codes with Testers}
\label{sec:concat_testers}

In general, the concatenation of codes admitting  $q$-testers
with positive soundness
does not share the same property.
However, in this section, we will show that under certain
assumptions, this is indeed the case.

\begin{dfn}\label{DF:compatibility}
Let $C\subseteq \Sigma^n$ be a code, let $T=(T^{(i)},(a^{(i)}_1,\dots,a^{(i)}_{q}),p_i)_{i\in I}$
be a $q$-tester for $C$ (see \S\ref{subsec:testers}),
and 
let $f:\Sigma\to \Delta^k$ be an injective function.
We denote the $j$-coordinate of $f$ as $f_j:\Sigma\to \Delta$.
We say that $T$ is \emph{$f$-compatible} if for every $i\in I$, there are $b_1  ,\dots,b_q  \in [k]$ and a function $g :\Delta^q\to \{0,1\}$
such that for every $w_1,\dots,w_q\in \Sigma$, 
\[
T^{(i)}(w_1,\dots,w_q)=g (f_{b_1}(w_1),\dots,f_{b_q}(w_q)).
\]
In the other words, we can compute $T^{(i)}(w_1,\dots,w_q)$
by reading   $q$ letters from
the word $f(w_1)\cdots f(w_q)$ in $\Delta^{qk}$.
\end{dfn}


\begin{thm}\label{TH:comp_of_codes_w_testers}
Let $C\subseteq \Sigma^n$ be a code admitting a $q_C$-tester 
$T_C$
with soundness $\mu_C>0$.
Let $f:\Sigma\to \Delta^k$ be an injective function such that the code $D:=\im(f)\subseteq \Delta^k$
has a $q_D$-tester $T_D$ with soundness $\mu_D>0$.
Suppose further that $T_C$ is $f$-compatible. Then the   concatenated code $E:=C\circ_f D\subseteq \Delta^{nk}$ admits a $\max\{q_C,q_D\}$-tester $T_E$ with soundness 
\[
\frac{\mu_C\mu_D}{qk+\mu_C+\mu_D}.
\]
\end{thm}

\begin{proof}
We abbreviate $q_C$ to $q$  and write $T=T_C=(T^{(i)},(a^{(i)}_1,\dots,a^{(i)}_{q}),p_i)_{i\in I}$.
Given $i\in I$, denote the $b_1  ,\dots,b_q  \in [k]$ and the  function $g :\Delta^q\to \{0,1\}$
promised by Definition~\ref{DF:compatibility} by 
$b_1^{(i)}  ,\dots,b_q^{(i)}  \in [k]$ and   $g^{(i)} :\Delta^q\to \{0,1\}$.

Fix some distribution $\rho=(\rho_1,\rho_2,\rho_3)$ on the set $\{1,2,3\}$;
we will specify $\rho$ at the end.
Using $\rho$,
we define a tester $T_E$ for $E\subseteq \Delta^{nk}$ as follows:
Given $u\in \Delta^{nk}$, write $u=u_1\cdots u_n$ with $u_1,\dots,u_n\in \Delta^k$.
Then choose $x\in\{1,2,3\}$ according to the distribution $\rho$
and perform routine $(x)$ from the following list.
\begin{enumerate}[label=(\arabic*)]
	\item Choose $\ell\in [n]$ uniformly at random and return $T_D(u_\ell)$
	($q_D$ letters are read from $u$).
	\item  Choose $i\in I$ according to the distribution $(p_i)_{i\in I}$
	and return $g^{(i)}\Circs{(u_{a_1^{(i)}})_{b_1^{(i)}},\dots,(u_{a_q^{(i)}})_{b_q^{(i)}}}$
	($q_C$ letters are read from $u$).
	\item  Choose $i\in I$ according to the distribution $(p_i)_{i\in I}$,
	choose $\ell\in [q_C]$ uniformly at random 
	and return $T_D(u_{a_\ell^{(i)}})$
	($q_D$ letters are read from $u$).
\end{enumerate}
The rationale behind (1) and (2) was explained in the introduction. 
Routine (3) is included in order to deal with situations where   $a_1^{(i)},\dots,a_q^{(i)}$ are very unevenly distributed; otherwise it is subsumed by routine (1).

It is straightforward to see that $T_E$ accepts every word in $E$. It remains to prove that
it has the claimed soundness.

Let $u\in \Delta^{nk}$ and write $\veps:=\dist(u,E)$. 
Then
\[
\Prob(T(u)=0)=\sum_{x=1}^3 \rho_x\cdot \Pr(T(u)=0\where \text{($x$) is executed}).
\]
We will bound the three summands on the right hand side from below.
To that end, we think of  $D^n$ as a subset of $\Delta^{nk}$
and  define the following words in $\Delta^{nk}$:
\begin{itemize}
	\item $u'$ is an element of $D^n$ with minimal distance from $u$.
	\item $u''$ is an element of $E$ with minimal distance from $u'$.
\end{itemize}
We further let $w' \in \Sigma^n$ denote the word  satisfying $f^n(w')=u'$. 
Obseve that
\begin{equation}\label{EQ:dist_sum}
\veps=\dist(u,E) \leq \dist(u,u'')\leq \dist(u,u')+\dist(u',u'').
\end{equation}

Now, since $T_D$ has soundness $\mu_D$, we have
\begin{align*}
\Prob(T_E(u)=0\where \text{(1) exec.})
& = \Prob_{\ell\sim [n]}\Circs{T_D(u_\ell)=0}
\geq \frac{1}{n}\sum_{\ell=1}^n\mu_D\dist_{\Delta^k}(u_\ell,D)
\\
&= \mu_D \dist_{\Delta^{nk}}(u,D^n)
= \mu_D\dist(u,u').
\end{align*}
Next, 
denote by $R$ the probability
that $u_{a_t^{(i)}} \in D$ for all $t=1,\dots,q$
(equivalently, $u_{a_t^{(i)}}=u'_{a_t^{(i)}}$ for $t=1,\dots,q$)
when $i\in I$ is chosen according to the distribution $p$.
Then
\begin{align*}
\Prob(T_E(u)= & \,0\where \text{(2) is exec.})
=
\Prob_{i\sim p}
\Squares{
g^{(i)}\Circs{(u_{a_1^{(i)}})_{b_1^{(i)}},\dots,(u_{a_q^{(i)}})_{b_q^{(i)}}}=0
}
\\
& \geq 
\Prob_{i\sim p}
\Squares{
g^{(i)}\Circs{(u_{a_1^{(i)}})_{b_1^{(i)}},\dots,(u_{a_q^{(i)}})_{b_q^{(i)}}}=0
~\text{and}~
u_{a_t^{(i)}}=u'_{a_t^{(i)}}~\forall\, t\in[q]
}  
\\
&=
\Squares{
g^{(i)}\Circs{(u'_{a_1^{(i)}})_{b_1^{(i)}},\dots,(u'_{a_q^{(i)}})_{b_q^{(i)}}}=0
~\text{and}~
u_{a_t^{(i)}}=u'_{a_t^{(i)}}~\forall\, t\in[q]
}  
\\
&=
\Prob_{i\sim p}\Squares{T^{(i)}(w'_{a_1^{(i)}},\dots,w'_{a_q^{(i)}})=0}-
\Prob_{i\sim p}\Squares{\exists t\in [q] \suchthat u_{a_t^{(i)}}\neq u'_{a_t^{(i)}} }
\\
& \geq  \mu_C \dist(w',C)-(1-R)
\\
& \geq  \mu_C\dist(u',E)-R=\mu_C\dist(u',u'')-(1-R)
\end{align*}
while
\begin{align*}
\Prob(T_E(u)= ~& 0\where \text{(3) is exec.})
=
\Prob_{i\sim p,  \ell\sim [q]}
(T_D(u_{a_\ell^{(i)}})=0)
\\
&
\geq (1-R) \cdot \Prob_{i\sim p,  \ell\sim [q]}
\Squares{
\left.
T_D(u_{a_\ell^{(i)}})=0\,\right|\, 
u_{a_\ell^{(i)}}\notin D~
\text{for some $\ell$}
}
\\
&\geq 
(1-R)\cdot \frac{1}{q}\cdot \frac{\mu_D}{k}
\end{align*}
Here, the second-to-last inequality holds since 
$u_{a_\ell^{(i)}}\notin D$
implies  $\dist_{\Delta^k}(u_{a^{(i)}_\ell},D)\geq \frac{1}{k}$.

Putting everything together gives
\begin{align*}
\Prob(T_E(u)=0)&\geq 
\rho_1\cdot \delta(u,u')\mu_D
+\rho_2\cdot(\mu_C\dist(u',u'')-(1-R))
+\rho_3\cdot (1-R)  \frac{\mu_D}{qk}
\\
&=\rho_1 \mu_D \delta(u,u')+\rho_2 \mu_C\dist(u',u'')+
(1-R)\Circs{-\rho_2+\rho_3\cdot \frac{\mu_D}{qk}}
=:(\star)
\end{align*}
Suppose that $\rho$ is chosen show that $\rho_2\leq \frac{\mu_D}{qk}\rho_3 $.
Then by \eqref{EQ:dist_sum}, we have
\[
(\star)\geq 
\rho_1 \mu_D \delta(u,u')+\rho_2 \mu_C\dist(u',u'')
\geq 
\veps \cdot \min \{\rho_1\mu_D,\rho_2\mu_C\}.
\]
Subject to the requirements $\rho_1+\rho_2+\rho_3=1$
and $\rho_2\leq \frac{\mu_D}{qk}\rho_3 $, the right hand side
is minimized for
\[\rho_1=\frac{\mu_C}{\mu_C+\mu_D+qk},
\quad \rho_2=\frac{\mu_D}{\mu_C+\mu_D+qk},
\quad \rho_3=\frac{qk}{\mu_C+\mu_D+qk}, 
\]
in which case it evaluates to $\frac{\mu_C\mu_D}{\mu_C+\mu_D+qk}\cdot\veps$,
which is what we want.
\end{proof}

\begin{remark}\label{RM:block_len_loss}
Without additional assumptions in Theorem~\ref{TH:comp_of_codes_w_testers},
the  tester $T_E$ described in its proof
cannot have   soundness larger than $\Omega(\frac{1}{k})$.
The reason is that we make no assumption on the distribution
of the coordinates  $b_1^{(i)},\dots,b_q^{(i)}\in [k]$ from Definition~\ref{DF:compatibility} as $i\in I$ distributes according to $p$.
For example, consider an extreme situation in which $b_1^{(i)}=\dots=b_q^{(i)}=1$ for all $i\in I$
and $\delta(D)>\frac{1}{k}$.
Start with a word $w'\in\Sigma^n$ that is $\delta$-far from $C$, let
$w\in C$ be the closest word to $w'$, and let $u=f(w_i)\cdots f(w_n)$ and
$u'_i=f(w'_i)\cdots f(w'_n)$. Then,
for every $i\in [n]$ such that $w_i\neq w'_i$,
replace the $(ki+1)$-letter of   $u'$ with the corresponding letter 
in $u$; denote the resulting word by $u''$. 
Now, the word $u''$ satisfies    $\delta(u'',C\circ D)\geq \delta(u',C\circ D)-\delta(u'',u')
\geq \delta\cdot (\delta(D)-\frac{1}{k})=\Omega(\delta)$. 
However,  
when we attempt to emulate  $T_C$ on $u''$ (routine (2) above), we cannot distinguish between
$u$ and $u''$. Thus, the only way we could detect that $u''$ is not in $C\circ D$
is by applying the tester of $D$ to one of the $k$-letter blocks of $u''$ (routines (1) or (3)).
Regardless of how this block is chosen,  it differs by only one letter
from a word in $D$, and so we are only guaranteed to succeed with probability
$\mu_D\cdot\frac{1}{k}$ or less. 
It follows that     $T_E$   rejects $u''$
with probability 
$O(\frac{1}{k})$, while  $u''$ is $\Omega(\delta)$-far from $C\circ D$.

One can overcome the $\Omega(\frac{1}{k})$ limitation  by requiring that for every
$t\in [q]$, the position $ka_t^{(i)}+b_t^{(i)}$  distributes nearly uniformly in $[nk]$ when $i\sim p$.
This requires imposing additional assumptions on $T$, $f$, $D$.
We omit the details because   these assumptions are not guaranteed
in the situations in which we apply Theorem~\ref{TH:comp_of_codes_w_testers}.

Another approach to bypass this limitation is to take into advantage the fact that in routine (3) above, we only need to check that the block $u_{a^{(i)}_t}$ was not corrupted \emph{in position $b^{(i)}_t$}. 
Calling $T_D(u_{a^{(i)}_t})$ is wasteful since it does not   use   this information.
To use it, however, we need  $D$ to be equipped with a ``local tester'' able to detect with high probability whether a given word is   far from $D$
\emph{or  corrupted in a given position}.
Such local testers were not considered in the literature, and are out of the scope of this work.
\end{remark}

Theorem~\ref{TH:comp_of_codes_w_testers} 
implies in particular that 
if $C\subseteq\Sigma^n$
admits a $q$-tester with soundness
$\mu>0$ and if $\Delta$
is an alphabet containing $\Sigma$,
then the code $C\subseteq \Delta^n$
(i.e., $C$ thought of as a code in $\Delta^n$)
has a $q$-tester with soundness $\frac{\mu}{q+1+\mu}$
(take $f:\Sigma\to \Delta^1$ to be the embedding of $\Sigma$
in $\Delta$).
In this case, however, there is a tester
with larger soundness.

\begin{prp}\label{PR:alphabet_increase}
Let $C\subseteq \Sigma^n$ be a code
admitting a $q$-tester $T$ with soundness $\mu>0$, and let $\Delta$
be an alphabet containing $\Sigma$.
Then $C\subseteq \Delta^n$ admits a tester $T'$
with soundness $\frac{\mu}{\mu+1}$.
\end{prp}

\begin{proof}
Write $T=(T^{(i)},(a^{(i)}_1,\dots,a^{(i)}_{q}),p_i)_{i\in I}$.
Fix $\rho_1,\rho_2>0$ with $\rho_1+\rho_2=1$
and consider the following tester $T'$ for $C\subseteq \Delta^n$:
Given $u\in \Delta^n$, choose $x\in\{1,2\}$ according to the distribution
$(\rho_1,\rho_2)$ and perform routine ($x$) from the following:
\begin{enumerate}
	\item[(1)] Choose $\ell\in [n]$ uniformly at random and accept
	$u$ if and only if $u_\ell\in \Sigma$.
	\item[(2)] Choose $i\in I$ according to $p$.
	Read $u_{a_1^{(i)}},\dots,u_{a_q^{(i)}}$.
	If one of these letters is not in $\Sigma$,
	then reject $u$.
	Otherwise, return $T^{(i)}(u_{a_1^{(i)}},\dots,u_{a_q^{(i)}})$.
\end{enumerate}

As in the proof of Theorem~\ref{TH:comp_of_codes_w_testers},
to show that $T'$ has the claimed soundness,
let $u'\in\Sigma^n$ be a word with minimal distance
from $u$, and let $u''\in C$ be a word with minimal distance
from $u'$.
Then
\[
\dist(u,u')+\dist(u',u'')\geq \dist(u,C)=:\veps.
\]
Moreover,
$\Prob(T'(u)=0\where \text{(1) exec.})=\dist(u,u')$.
Next, as in the proof of Theorem~\ref{TH:comp_of_codes_w_testers}, writing $R$ for the probability
that $u_{a_1^{(i)}},\dots,u_{a_q^{(i)}}\in \Sigma$
when $i\in I$ distributes according to $p$, 
we have
\begin{align*}
\Prob(T'(u)=0\where \text{(2) is exec.})
& = 
(1-R) + \Prob_{i\sim p}\Squares{T^{(i)}(u_{a_1^{(i)}},\dots,u_{a_q^{(i)}})=0~\text{and}~
u_{a_1^{(i)}},\dots,u_{a_q^{(i)}}\in \Sigma}
\\
&\geq  (1-R)+[\mu\dist(u',C)-(1-R)]\geq \mu\dist(u',u'').
\end{align*}
It follows
that
\[
\Prob(T'(u)=0)\geq \rho_1 \dist(u,u')+\rho_2 \mu \dist(u,u'')
\geq \min\{\rho_1,\rho_2\mu\}\veps.
\]
Taking $\rho_1=\frac{\mu}{\mu+1}$ and $\rho_2=\frac{1}{\mu+1}$
completes the proof of proposition.
\end{proof}

\begin{remark}\label{RM:linearity_of_testers}
	In Theorem~\ref{TH:comp_of_codes_w_testers}, if
	$\Sigma$ and $\Delta$ are $\F$-vector spaces,
	$f:\Sigma\to \Delta$ is $\F$-linear,
	and the code $C$ is an $\F$-code,
	then $E=C\circ D\subseteq \Delta^{nk}$ is also an $\F$-code.
	Moreover, it is clear from the proof 
	that if the testers $T_C$ and $T_D$ are $\F$-linear,
	then so is the tester $T_E$ of $E$.
	
	Similarly, in Proposition~\ref{PR:alphabet_increase}, if
	$\Sigma$ and $\Delta$
	are $\F$-vector spaces such that $\Sigma$
	is a subspace of $\Delta$
	and $T$ is   $\F$-linear, then so is the tester $T'$ for $C\subseteq \Delta^n$.
\end{remark}

\section{The Generalized Hadamard Code and the Generalized Long Code}
\label{sec:gen}

In this section, we introduce two types of 
codes generalizing the (linear) Hadamard code
and the (non-linear) long code, respectively.
They will  ultimately serve as the inner code  when we apply Theorem~\ref{TH:comp_of_codes_w_testers},
and to that end, we will show that
they  can be defined using
$2$-letter constraints.
This stands in contrast to the original
Hadamard code and long code which, as we also show, 
cannot be defined by $2$-letter constraints.

\subsection{A General Construction}
\label{subsec:general_setting}

We begin with a general construction which
is natural in the context of Theorem~\ref{TH:comp_of_codes_w_testers}. 
Let $S$ be a finite set, let $\Delta$
be an alphabet and let $f_1,\dots,f_n:S\to\Delta$
be some functions.
We associate with $S$, $\Delta$ and the collection $\{f_i\}_{i=1}^n$
the code
\[
D(\{f_i\}_{i=1}^n):=\{(f_1(s),\dots,f_n(s))\where s\in S\}\subseteq \Delta^n.
\]
If we further define
$f:S\to \Delta^n$ by  $f(s)=(f_1(s),\dots,f_n(s))$ and assume that $f$
is injective, then this recovers the setting
of Theorem~\ref{TH:comp_of_codes_w_testers}.
In fact, every code $D\subseteq \Delta^n$ equipped with an encoding function $f:S\to D$
is of the form $ D(\{f_i\}_{i=1}^n)$, but the point of our construction
is to build a code  in $\Delta^n$ from   functions $f_1,\dots,f_n:S\to \Delta$ 
rather than doing the opposite.

The construction  recovers
two well-known codes.
When $S=\{0,1\}^k$, $\Delta=\{0,1\}$
and $f_1,\dots,f_n$ consist of all the functions from $S$ to $\{0,1\}$
(so that $n=2^{2^k}$),
the code $C$ is known as the \emph{long code}, e.g., see \cite[\S3]{Bellare_1998_free_bits_PCPs}.
Next, when $S=V$ with $V$  an $\F$-vector space, $\Delta=\F$
and $f_1,\dots,f_n$ are all the $\F$-linear functions
from $V$ to $\F$ (so that $n=|V|$), the code
$C$ is  (up to equivalence)  the $|\F|$-ary \emph{Hadamard code}
associated to the vector space $V$, e.g., see \cite[Exercise~2.22]{Guruswami_2025_ess_coding_thy_preprint} (we also recall the definition in \S\ref{subsec:gen_Hadamard}).

We now introduce a $q$-tester for
$ D(\{f_i\}_{i=1}^n)$.
To that end, we say that
a collection of functions $f_1,\dots,f_q:S\to \Delta$
is \emph{dependent}
if
\[
\im(f_1\times \dots\times f_q)=\{(f_1(s),\dots,f_q(s))\where s\in S\}\subsetneq \Delta^q.
\]
Informally, this means that knowing (say) the values of
$f_1,\dots,f_{q-1}$ on some input $s\in S$  gives some
information on the value of $f_q$ on that input, even without knowing
what $s$ is.
(Note that it is possible for a single function $f_1$ to be dependent, namely,
if its image is strictly smaller than $\Delta$.)
With this   at hand, we introduce:

\begin{dfn}[$q$-dependence tester]
	\label{DF:dep_tester}
	With notation as a above, given a word $w\in \Delta^n$,
	the \emph{$q$-dependence tester}
	for $ D(\{f_i\}_{i=1}^n)$ chooses uniformly at random  
	a tuple
	$(i_1,\dots,i_q)\in [n]^q$  such that $f_{i_1},\dots,f_{i_q}$
	are dependent, and accepts $w$
	if and only if
	$(w_{i_1},\dots,w_{i_q})\in \im(f_{i_1}\times \dots\times f_{i_q})$.
\end{dfn}

Clearly, if the $q$-dependence tester has positive soundness,
then this soundness is at least $n^{-q}$.

\begin{example}\label{EX:dep_tester}
To demonstrate how the $q$-dependence tester works, suppose for simplicity
that $\Delta=\F$, and let $w\in \Delta^k$
be a word that is always accepted by the $3$-dependence tester.
We claim that if $f_i,f_j,f_k$ satisfy $f_i+f_j=f_k $,
then we must have $w_i+w_j=w_k$.
Indeed, this is forced by the fact that $\im(f_i\times f_j\times f_k)\subseteq\{(\alpha,\beta,\alpha+\beta)\where
\alpha,\beta\in \F\}$.
Similarly, if $f_i=\alpha f_j$ for some $\alpha\in \F$,
then $w_i=\alpha w_j$, because for any $k\in [n]$,
$\im(f_i,f_j,f_k)\subseteq\{(\beta,\alpha\beta,\gamma)\where
\beta,\gamma\in \F\}$ ($f_k$ was added artificially to $f_i$ and $f_j$
to make a dependent triple of functions).
More generally, for a general $\Delta$,
if there is a function $g:\Delta\times \Delta\to \Delta$
such that $g\circ (f_i\times f_j)=f_k$ (i.e.\ $g(f_i(s),f_j(s))=f_k(s)$
for all $s\in S$),
then $g(w_i,w_j)=w_k$.
\end{example}

The following useful property of the $q$-dependence tester
will be used freely in the sequel.

\begin{prp}\label{PR:univ_depen_tester}
With notation as above, the code $D(\{f_i\}_{i=1}^n)\subseteq \Delta^n$
is defined by $q$-letter constraints (i.e., it has a $q$-tester with positive soundness)
if and only if its $q$-dependence tester has positive soundness.
\end{prp}

\begin{proof}
The ``if'' direction is clear, so we treat the ``only if'' direction.
Let $T=(T^{(i)},(a_{1}^{(i)},\dots,a_{q}^{(i)}),p_i)_{i\in I}$
be a $q$-tester for $D:=D(\{f_i\}_{i=1}^n)$ having positive soundness.
Let $w\in \Delta^n-D$. It is enough to prove that $w$ fails the $q$-dependence
test with some positive probability. Since $T$ has positive soundness,
there is $i\in I$ such that $T^{(i)}(w_{a_1^{(i)}},\dots,w_{a_q^{(i)}})=0$.
Since $T$ accepts all words in $D$, we must also
have $T^{(i)}(f_{a_1^{(i)}}(s),\dots,f_{a_q^{(i)}}(s))=1$ for all $s\in S$.
Consequently, $(w_{a_1^{(i)}},\dots,w_{a_q^{(i)}})\notin
\im(f_{a_1^{(i)}}\times \dots \times f_{a_q^{(i)}})$.
This means that $f_{a_1^{(i)}},\dots, f_{a_q^{(i)}}$ are dependent
and that $w$ will be rejected if the $q$-dependence test chooses these functions.
\end{proof}

\subsection{The Generalized Hadamard Code}
\label{subsec:gen_Hadamard}

Keeping the notation above,
suppose that both $S$
and $\Delta$
are $\F$-vector spaces.
We write $V$ in place of $S$ to denote that.
Let  $f_1,\dots,f_n$ enumerate all the $\F$-linear functions from 
$V$ to $\Delta$ (so that $n=|\F|^{\dim V\cdot \dim \Delta}=|\Delta|^{\dim V}=|V|^{\dim \Delta}$). 
The \emph{generalized Hadamard code} associated
to $V$ and $\Delta$ is defined to be
\[
H(V,\Delta):=D(\{f_i\}_{i=1}^n)=\{(f_1(v),\dots,f_n(v))\where v\in V\}\subseteq \Delta^n.
\]
It easy to see that $H(V,\Delta)\subseteq\Delta^n$
is an $\F$-code
of relative distance $1-\frac{1}{|\Delta|}$
and rate $\frac{\log_{|\Delta|}|V|}{|\Delta|^{\dim V}}=\frac{\dim V}{\dim\Delta\cdot|\Delta|^{\dim V}}=\frac{\log_{|\Delta|}n}{n\dim \Delta}$.

\begin{remark}
Recall that the Hadamard code associated to an $\F$-vector space $V$
is   
$\Hom_{\F}(V,\F)$ ---
the set of linear
functions from $V$ to $\F$ --- viewed as a subset of $\F^V$, i.e.,
all functions from $V$ to $\F$. More concretely, writing $v_1,\dots,v_n$ for the vectors in $V$,
the Hadamard code of $V$
is
\[
H'=\{(f(v_1),\dots,f(v_n))\where \text{$f:V\to \F$ is linear}\}\subseteq \F^n.
\]
This code is in fact equivalent to $H(V,\F)$,
which is the reason why we call $H(V,\Delta)$ a generalized Hadamard code.
Indeed, let $V^*$ denote the dual vector space of $V$.
Then $f_1,\dots,f_n$ are the vectors of $V^*$.
Since every linear map $g$ from $V^*$ to $\F$ admits a unique
$v\in V$
such that   $g(f)=f(v)$  for all $f\in V^*$,
the code $H'$ is just $H(V^*,\F)$. Choosing an isomorphism $V\to V^*$
then gives an equivalence between $H'$ and $H$.
\end{remark}

It is well-known that the code $H(V,\F)$ is $3$-testable.
Indeed, it is equivalent to the ordinary 
Hadamard code $\Hom_\F(V,\F)\subseteq \F^V$,
which is defined by $3$-letter constraints, namely,
\begin{align*}
\Hom_\F(V,\F)=\{f:V\to \F\suchthat & \text{$f(x+y)=f(x)+f(y)$ and} 
\\
& \text{$f(\alpha x)=\alpha f(x)$ for all $x,y\in V$, $\alpha\in \F$}\}.
\end{align*}
Tracing this back to $H(V,\F)\subseteq \F^n$ ($n=|V|$) shows that
a word $w\in \F^n$ is in $H(V,\F)$ if and only if
\begin{enumerate}[label=(\arabic*)]
		\item $w_i=w_j+w_k$ for all $i,j,k\in [n]$ such that $f_i=f_j+f_k$ and
		\item $w_i=\alpha w_j$ for all $i,j\in [n]$ and $\alpha\in \F$
		such that $f_i=\alpha f_j$.
\end{enumerate}
Consequently, the $3$-dependence tester of 
$H(V,\F)$ has positive soundness  (cf.\ Example~\ref{EX:dep_tester}).
(In fact, for some fields $\F$, e.g.\ $ \F_2$,
the Hadamard code is known to have $3$-testers with soundness that
is independent of $V$; see \cite{Blum_1993_self_testing}, \cite{Bellare_1996_linearity_testing_char_two} and the many followup works.)
On the other hand, it is not difficult to see that
the $2$-dependence tester for $H(V,\F)$ has soundness $0$ if $\dim V>1$; the reason
is that a pair of $\F$-linear functions
$f,g:V\to \F$ is dependent if and only if they are
proportional. As a result,
$H(V,\F)$ is not defined by $2$-letter constraints
(Proposition~\ref{PR:univ_depen_tester}).
By contrast, we  now show that  $H(V,\Delta)$ 
is defined by $2$-letter constraints when $\dim \Delta\geq 2$.

\begin{thm}\label{TH:gen_Hadamard_2_test}
	Let $\Delta$ be a vector space of dimension $>1$
	and let $V$ be any vector space.
	Then the $2$-dependence tester (Definition~\ref{DF:dep_tester})
	for the code $H(V,\Delta)\subseteq \Delta^n$
	has positive soundness (depending on $n$).
\end{thm}

Let us first explain the idea of the proof for $\Delta=\F^2$.
The generalized Hadamard code
$H(V,\Delta)$ is defined by   the same kind of $3$-letter
constraints (1), (2) above. However, when $\Delta=\F^2$, one can
simulate such  $3$-letter constraints using only $2$ letters by packing pairs of scalar values into a single vector in $\F^2$.
Concretely, if $g_1, g_2, g_3 : V \to \mathbb{F}$ are some  linear functionals with
$g_1 + g_2 = g_3$, then there is some $\ell\in[n]$
such that   
$f_\ell=(g_1,g_2):V\to \Delta=\F^2$, that is,
$ f_\ell(v)= (g_1(v), g_2(v) ) $. This single function
is  dependent
with each of the functions $ (g_1, 0 ) ,   (g_2, 0 ),(g_3, 0 ):V\to \F^2$. Given a word $w\in\Delta^n$ passing the
$2$-dependence test,
the proof uses such dependencies to first recover from  $w$ a vector
$v\in V$ such that whenever $f_i:V\to \Delta$ has the form $(g ,0)$,
we have $w_i= (g (v),0 )$. Then, to show that $f_j(v)=w_j$ for general $j$,
one writes $f_j=(g_1,g_2)$ and uses the dependency between
$f_j$ and each of the functions $(g_1,0)$, $(g_2,0)$.

\begin{proof}[Proof of Theorem~\ref{TH:gen_Hadamard_2_test}]
	We may assume without loss of generality
	that $\Delta=\F^m$ for some $m>1$.
	Given functions   $g_1,\dots,g_r\in \Hom_{\F} (V,\F)$ ($r\leq m$),
	we denote by $(g_1,\dots,g_r,\vec{0})$
	the function from $V$ to $\Delta=\F^m$
	sending $v\in V$ to $(g_1(v),\dots,g_r(v),0,\dots,0)\in \F^m=\Delta$.
	Let   $g_1,\dots,g_t$ enumerate all
	the linear functions from $V$ to $\F$ ($t=|V|$).
	We may assume that $f_1,\dots,f_n$
	are numbered so that $f_i=(g_i,\vec{0})$ whenever
	$1\leq i\leq t$.

	Let $w\in \Delta^n$ be a word that is always accepted
	by the $2$-dependence tester
	of $H(V,\Delta)$. We need to show that $w\in H(V,\Delta)$.

	Since $\im (f_1),\dots,\im (f_t)\subseteq \F\times\{0\}^{m-1}$
	and $w$ passes the $2$-dependence test,
	we must have $w_1,\dots,w_t\in \F \times\{0\}^{m-1}$.
	Thus, for each $i\in [t]$, we can write $w_i=(u_i,0,\dots,0)$
	with $u_i\in\F$.
	
	Let
	$u=u_1\cdots u_t$.
	We claim that  $u\in D(\{g_i\}_{i=1}^t)= H(V,\F)$.
	As we noted earlier, in order to show this, it
	is enough to check that for all $i,j,k\in [t]$
	with $g_i+g_j=g_k$ (equiv.\ $f_i+f_j=f_k$),
	we have $u_i+u_j=u_k$,
	and for all $i,j\in [t]$ and $\alpha\in \F$
	with $g_i=\alpha g_j$ (equiv.\ $f_i=\alpha f_j$),
	we have $u_i=\alpha u_j$.
	The second statement holds because $w$ 
	passes the $2$-dependence test,
	so we only need to establish the first.
	There is $\ell\in [n]$ such that $f_\ell=(g_i,g_j,\vec{0})$.
	Then $\im (f_i\times f_\ell)\subseteq \{((\alpha,0,\dots,0),(\alpha,\beta,0,\dots,0))\where \alpha,\beta\in\F\} $,
	and so our assumption on $w$ implies that $w_\ell$
	has the form $(u_i,*,0,\dots,0)$.
	Repeating this argument with $f_j$ in place of $f_i$
	gives $w_\ell=(u_i,u_j,0,\dots,0)$.
	Now, since 	 $\im (f_k\times f_\ell)\subseteq \{((\alpha+\beta,0,\dots,0),(\alpha,\beta,0,\dots,0))\where
	\alpha,\beta\in\F\}$, we must have $u_k=u_i+u_j$, which is what we want.
	We conclude that $u\in H(V,\F)$.
	
	Let $v\in V$ be a vector such that $u_i=g_i(v)$ for all $i\in [t]$.
	We finish the proof by showing that $w_i=f_i(v)$ for all $i\in [n]$.
	Let $i\in [n]$. Then there are $i_1,\dots,i_m\in[t]$
	such that $f_i=(g_{i_1},\dots,g_{i_m})$.
	Now,  for every $\ell\in [m]$, 
	the dependence between $f_i$ and $f_{i_\ell}$ implies
	that
	the $\ell$-th component of $w_i\in\F^m$ is $u_{i_\ell}=g_{i_\ell}(v)$.
	As this holds for every $\ell$,
	it follows that $w_i=(g_{i_1}(v),\dots,g_{i_m}(v))=f_i(v)$,
	as required.
\end{proof}

\subsection{The Generalized Long Code}
\label{subsec:gen_long_code}

Returning to the general setting of \S\ref{subsec:general_setting},
let $f_1,\dots,f_n$ denote all functions from $S$
to $\Delta$ (so that $n=|\Delta|^{|S|}$). We call
\[
L(S,\Delta):=D(\{f_i\}_{i=1}^n)=\{(f_1(s),\dots,f_n(s))\where s\in S\}\subseteq\Delta^n
\]
the \emph{generalized long code} associated to $S$ and $\Delta$.
This code has relative distance $1-\frac{1}{|\Delta|}$
and rate $\frac{\log_{|\Delta|}|S|}{|\Delta|^{|S|}}=\frac{\log_{|\Delta|}\log_{|\Delta|} n}{n}$.
As we noted earlier, this is a generalization of the \emph{long code},
which arises as the special case where $S=\{0,1\}^k$ and $\Delta=\{0,1\}$.

We will show in Remark~\ref{RM:long_code_not_2_test} that $L(S,\Delta)$ cannot
be defined by $2$-letter constraints when $|\Delta|=2$.
However, it turns out that it can be defined by $2$-letter
constraints when $|\Delta|>2$.

\begin{thm}\label{TH:gen_long_code_2_test}
	Let $S$ be a set and let $\Delta$ be an alphabet with more than $2$ letters.
	Then the $2$-dependence tester of the generalized long code 
	$L(S,\Delta)\subseteq \Delta^n$ has positive soundness (depending on $n$).
\end{thm}

In order to prove the theorem, we first note that $L(S,\{0,1\})$
has a $3$-tester with positive soundness. This is well-known for $S=\{0,1\}^k$, e.g., see 
\cite[Prp.~3.2]{Bellare_1998_free_bits_PCPs},
and  we give here a proof for general $S$.

\begin{prp}\label{PR:long_code_3_test}
With notation as above, suppose that $\Delta=\F_2$ 
and the functions $f_1,\dots,f_n$ from $S$
to $\Delta$
are numbered so 
that $f_1$ is the function sending every element of $S$ to $1$.
Then a word $w\in \F_2^n$ is in $L(S,\Delta)$ if and only if
\begin{enumerate}[label=(\arabic*)]
	\item $w_i +w_j= w_k$ (in $\F_2$) for every $i,j,k\in [n]$ such that $f_i+f_j=f_k$,
	\item $w_i\cdot w_j=w_k$ for every $i,j,k\in [n]$
	such that $f_i\cdot f_j=f_k$, and
	\item $w_1=1$. 
\end{enumerate}
In particular, $L(S,\F_2)$
is defined by $3$-letter constraints.
\end{prp}

\begin{proof}
	The ``only if'' part is clear, so we turn to prove the ``if''
	part. 
	The functions $f_1,\dots,f_n:S\to \F_2$
	are the elements of the $\F_2$-algebra $\F_2^S$ of functions from $S$ to $\F_2$.
	Let us think of $w$ as a function from $\F_2^S$ to $\F_2$  mapping $f_i$ to $w_i$
	for all $i\in [n]$.
	Assumptions (1)--(3) mean that $w$ is a ring homomorphism (respecting the unity)
	from $\F_2^S$
	to $\F_2$. In particular, its kernel is a maximal ideal $M$ of $\F_2^S$.
	From the structure theory of products of rings, 
	$M=\prod_{s\in S}I_s$, where each $I_s$ is an ideal of the ring $\F_2$.
	Since $M$ is maximal, there must be some $s\in S$ such that
	$I_s=0$ and $I_{s'}=\F_2$ for all $s'\neq s$,
	or equivalently,	
	$M=\{f\in \F_2^S\suchthat f(s)=0\}$.
	Now, for every $i\in [n]$, either $f_i\in M=\ker(w)$, and then $w_i=0=f_i(s)$,
	or $f_i\notin M=\ker(w)$, and then $w_i=1=f_i(s)$. We conclude that
	$w_i=f_i(s)$ for all $s\in S$, and thus $w\in L(S,\Delta)$.
\end{proof}

We shall also need the following key lemma.

\begin{lem}\label{LM:critical_lemma}
	Let $g_1,\dots,g_t:S\to\{0,1\}$ be distinct functions.
	For every $i,j\in [t]$, let $g_{i,j}$ and $g'_{i,j}$
	denote the functions from $S$ to $\{0,1,2\}$
	defined as follows:
	\[
	g_{i,j}(s)=
	\left\{
	\begin{array}{ll}
	0 & f_i(s) = 0\\
	1 & f_i(s)=1 \wedge f_j(s)=0 \\
	2 & f_i(s)=1 \wedge f_j(s)=1
	\end{array}
	\right.
	\qquad
	g'_{i,j}(s)=
	\left\{
	\begin{array}{ll}
	1 & f_i(s) = 1\\
	0 & f_i(s)=0 \wedge f_j(s)=1 \\
	2 & f_i(s)=0 \wedge f_j(s)=0
	\end{array}
	\right.
	\]
	Let $f_1,\dots,f_n:S\to \{0,1,2\}$
	enumerate the functions in $\{g_i,g_{i,j},g'_{i,j}\where i,j\in [t]\}$.
	If the code $D':=D(\{g_i\}_{i=1}^t)\subseteq \{0,1\}^t$ is defined
	by $3$-letter constraints,
	then   $D:=D(\{f_j\}_{j=1}^n)\subseteq \{0,1,2\}^n$
	is defined by $2$-letter constraints.
\end{lem}

\begin{proof}
	We may assume that $f_1,\dots,f_n$
	are chosen so that $f_i=g_i$ when $i\in [t]$.
	
	Suppose $w\in \{0,1,2\}^n$
	always passes the $2$-dependence test of $D$.
	We need to show that $w\in D$.
	We begin by noting that since $\im(f_i)\subseteq \{0,1\}$
	for all $i\in [t]$, we must have $w_i\in\{0,1\}$
	for $i\in [t]$.
	
	Next, fix $i,j\in [t]$ and let $k,\ell\in [n]$
	be the numbers for which $f_k=g_{i,j}$ and $f_\ell=g'_{i,j}$.
	Observe that $f_i=g_i$ and $f_k=g_{i,j}$ are dependent,
	and $f_j=g_j$ and $f_k=g_{i,j}$
	are dependent. These dependencies and our assumption on $w$
	imply that (i)   $w_k=0$ if and only if $w_i=0$,
	(ii) if $w_k=1$, then $w_j=0$, and (iii)
	if $w_k=2$, then $w_j=1$. This means that
	\begin{equation}\label{EQ:w_k_dep}
	w_k=
	\left\{
	\begin{array}{ll}
	0 & w_i = 0\\
	1 & w_i=1 \wedge w_j=0 \\
	2 & w_i=1 \wedge w_j=1.
	\end{array}
	\right.
	\end{equation}
	Arguing similarly with $\ell$ in place of $k$ gives
	\begin{equation}\label{EQ:w_l_dep}
	w_\ell=
	\left\{
	\begin{array}{ll}
	1 & w_i = 1\\
	0 & w_i=0 \wedge w_j=1 \\
	2 & w_i=0 \wedge w_j=0.
	\end{array}
	\right.
	\end{equation}
	
	At this point, if we could show that $w':=w_{1}\cdots w_t$ is
	in $D'$, then we  could conclude that $w\in D$.
	Indeed, if it were the case that $w'\in D'$, then there would
	be an $s\in S$ such that
	$w_i = g_i(s)= f_i(s)$ for all $i\in [t]$.
	Since for every	
	$p\in [n]-[t]$, the function $f_p$ is either $g_{i,j}$ or $g'_{i,j}$
	for some $i,j\in [t]$, by  the previous paragraph, we would have
	that 
	$w_p$ is   $g_{i,j}(s) $ or $g'_{i,j}(s)$, respectively,
	and 
	consequently $w_p=f_p(s)$. 
	
	Suppose for the sake of contradiction that
	$w'\notin D'$. 	
	By assumption, $D'$ admits a $3$-tester
	$T=(T^{(u)},(a_u,b_u,c_u),p_u)_{u\in U}$ with positive
	soundness. Then there is some $u\in U$ such that
	$T^{(u)}(w_{a_u},w_{b_u},w_{c_u})=0$.
	Write $(i,j,p)=(a_u,b_u,c_u)$
	and $(x_1,x_2,x_3)=(w_{a_u},w_{b_u},w_{c_u})$.
	By the pigeonhole principle, at least two of $x_1,x_2,x_3$
	are equal. We  reorder $(a_u,b_u,c_u)$
	to have $x_1=x_2$.
	Suppose that $x_1=x_2=1$.
	The fact that $T^{(i)}(v_i,v_j,v_p)=1$
	for all $v\in D'$ while $T^{(i)}(x_1,x_2,x_3)=0$
	means that 
	\[\im(f_i\times f_j\times f_p)\subseteq \{0,1\}^3-\{(x_1,x_2,x_3)\}
	=
	\{0,1\}^3-\{(1,1,x_3)\}.
	\]
	As a result, $\im(g_{i,j}\times f_p)\subseteq \{0,1,2\}\times\{0,1\}-\{(2,x_3)\}$.
	Writing $g_{i,j}=f_k$ with $k\in [n]$,
	this means that $f_k$ and $f_p$
	are dependent (as functions from $S$ to  $\{0,1,2\}$),
	and thus $(w_{k},w_p)\in  \{0,1,2\}\times\{0,1\}-\{(2,x_3)\}$.
	However, since $w_i=x_1=1$ and $w_j=x_2=1$, we
	have  $w_k=2$ by \eqref{EQ:w_k_dep}, and since $w_p=x_3$,
	it follows that $(w_k,w_p)=(2,x_3)$, a contradiction to
	our earlier conclusion.
	Similarly, when $x_1=x_2=0$,
	we reach a contradiction using \eqref{EQ:w_l_dep}.
	As all possibilities lead to contradiction,
	our assumption $w'\notin D'$ must have been false.
	This completes the proof.
\end{proof}

We are now ready to prove Theorem~\ref{TH:gen_long_code_2_test}.

\begin{proof}[Proof of Theorem~\ref{TH:gen_long_code_2_test}]
	We may assume without loss of generality that $\{0,1,2\}\subseteq \Delta$
	and that the functions
	$f_1,\dots,f_n:S\to \Delta$
	are numbered so that $f_1,\dots,f_t$ ($t=2^{|S|}$)
	are all the functions from $S$ to $\{0,1\}$.
	
	Suppose that $w\in \Delta^n$ always passes the $2$-dependence test.
	We need to show that $w\in L(S,\Delta)$.
	
	We begin by showing that the exists $s\in S$
	such that $w_i=f_i(s)$ for all $i\in [t]$.
	To that end, let $g_i=f_i$ for $i\in [t]$,
	and for all $i,j\in [t]$,
	define   $g_{i,j}, g'_{i,j}:S\to \{0,1,2\}$
	as in Lemma~\ref{LM:critical_lemma}.
	We may assume that there is $m\in\{t,\dots,n\}$
	such that $f_1,\dots,f_m$
	enumerate all the functions in the set $\{g_i,g_{i,j},g'_{i,j}\where i,j\in [t]\}$.
	Now, since $w$ passes the $2$-dependence,
	and since $\im(f_i)\subseteq \{0,1,2\}$ for all $i\in [m]$,
	we have $w_i\in \{0,1,2\}$ for all $i\in [m]$.
	We may therefore think of $w':=w_1\cdots w_m$
	as a word in $\{0,1,2\}^m$.
	Our assumption on $w$ implies that $w'$
	also passes the $2$-dependence test
	for the code $D(\{f_i\}_{i=1}^m)$,
	so by Lemma~\ref{LM:critical_lemma} and Proposition~\ref{PR:long_code_3_test} (applied
	to $D(\{f_i\}_{i=1}^t)=L(S,\{0,1\})$),
	we have
	$w'\in D(\{f_i\}_{i=1}^m)$. This means in particular
	that there is $s\in S$ such
	that $w_i=f_i(s)$ for all $i\in [t]$ (and even for all
	$i\in [m]$).
	
	We finish by showing that $w_i=f_i(s)$ for all $i\in [n]$.
	Fix some $i\in [n]$. 
	For every $a\in \Delta$, let $g_a:\Delta\to \{0,1\}$
	be the function mapping $a$ to $1$ and all other elements
	of $\Delta$ to $0$.
	There is some $j\in [t]$ such that $f_j=g_a\circ f_i$.
	The dependency between $f_i$ and $f_j$
	and our assumption on $w$ now imply that
	$w_j=g_a(w_i)$ (cf.\ Example~\ref{EX:dep_tester}).
	Since $w_j=f_j(s)$, we get
	that $g_a(w_i)=f_j(s)=g_a(f_i(s))$.
	By the definition of $g_a$, this means
	that   $w_i=a$ if and only if $f_i(s)=a$.
	As this holds for every $a\in \Delta$,
	we conclude that $w_i=f_i(s)$.
\end{proof} 

\begin{remark}\label{RM:long_code_not_2_test}
	By contrast to
	Theorem~\ref{TH:gen_long_code_2_test},	
	the code $L(S,\{0,1\})\subseteq \{0,1\}^{n}$ ($n=2^{|S|}$)
	is not defined by $2$-letter constaints when $|S|>2$.
	To show this, it is enough to exhibit a word $w\in \{0,1\}^{n}$
	that always passes the $2$-dependence test and is not in $L(S,\{0,1\})$.
	Fix   distinct $x_1,x_2,x_3\in S$  and  
	define the word $w\in \{0,1\}^n$ by  
	\[
	w_i=\maj(f_i(x_1),f_i(x_2),f_i(x_3))\qquad\forall i\in [n],
	\]
	where $\maj(-)$ is the majority function.
	The word $w$ is not in $L(S,\{0,1\})$
	because for every $s\in S$, there is $i\in [n]$
	such that $f_i(s)\neq \maj(f_i(x_1),f_i(x_2),f_i(x_3))=w_i$.
	
	Let us now show that $w$ always passes the $2$-dependence test. 
	Suppose that $f_i,f_j:S\to\{0,1\}$ are dependent.
	We need to show that $(w_i,w_j)\neq (a,b)$ for every $(a,b)\in \{0,1\}^2-\im(f_i\times f_j)$. Fix such $(a,b)$.
	If $w_i\neq a$, then it is clear that $(w_i,w_j)\neq (a,b)$. 
	On the other hand, if $w_i=a$, then $ \maj(f_i(x_1),f_i(x_2),f_i(x_3))=a$, and so at least $2$
	of $f_i(x_1),f_i(x_2),f_i(x_3)$ equal $a$. Since $(a,b)\notin
	\im(f_i\times f_j)$, at least $2$ of  $f_j(x_1),f_j(x_2),f_j(x_3)$
	are different from $b$. This means that $w_j\neq b$ and again
	we get $(w_i,w_j)\neq (a,b)$, as required.

	When $|S|=2$, the long code $L(S,\{0,1\})$ is 
	$\{0101,0011\}\subseteq \{0,1\}^4$ (up to equivalence), and  it is easy to see	that it is defined by $2$-letter constaints.
\end{remark}

\section{Separable Testers}
\label{sec:sep}

In order to apply Theorem~\ref{TH:comp_of_codes_w_testers}
and change the alphabet of an LTC from $\Sigma$ to $\Delta$,
we need to have a suitable function $f:\Sigma\to \Delta^k$
such that the tester $T$ of our code is $f$-compatible
(Definition~\ref{DF:compatibility}).
In this section, we give a necessary and sufficient condition
on $T$ to be compatible with \emph{some} 
$f$, and moreover show that every tester can be replaced
with one that satisfies our condition and has proportional soundness.

\begin{dfn}
Let $C\subseteq \Sigma^n$ be a code and let $T=(T^{(i)},(a^{(i)}_1,\dots,a^{(i)}_q),p_i)_{i\in I}$ be a $q$-tester  for $C$.
Let $\Delta$ be another alphabet. We
say that the tester $T$ is \emph{$\Delta$-separable} if for every $i\in I$,
there are functions $g_1,\dots,g_q:\Sigma\to \Delta$ such that
$T^{(i)}:\Sigma^q\to \{0,1\}$ factors via the function
$g_1\times\dots\times g_q:\Sigma^q\to\Delta^q$,
i.e., there is $g:\Delta^q\to\{0,1\}$
such that $T^{(i)}=g\circ (g_1\times\dots\times g_q)$.

When $C$ is an $\F$-code, $T$ is $\F$-linear, and $\Delta$ is an $\F$-vector space, we say
that $T$ is \emph{linearly $\Delta$-separable} if $g_1,\dots,g_q$
can be taken to be $\F$-linear maps.
\end{dfn}

\begin{example}
	Suppose that $C\subseteq \Sigma^n$
	is an $\F$-code and 
	$T=(T^{(i)},(a^{(i)}_1,\dots,a^{(i)}_q),p_i)_{i\in I}$ is
	$q$-tester such that  
	each $T^{(i)}$ checks
	a single linear constraint involving at most $q$ letters.
	Then $T$ is linearly $\F$-separable.
	Indeed, by assumption,
	for every $i\in I$,
	there is a linear map $\vphi_i : \Sigma^q\to \F$
	such that for every $w\in \Sigma^q$,
	we have $T^{(i)}(w)=1$ if and only if $\vphi_i(w)=0$.
	For such $\vphi_i$, 
	there are linear functionals
	$g_1,\dots,g_q:\Sigma\to \F$
	such that $\vphi_i(w)=g_1(w_1)+\dots+g_q(w_q)$.
	In particular, we can determine from $g_1(w_1),\dots,g_q(w_q)$
	the value of $T^{(i)}(w)$, meaning that $T^{(i)}$
	factors via $g_1\times\dots\times g_q:\Sigma^q\to\F^q$.
\end{example}

We now show that $\Delta$-separability is a necessary and sufficient
condition for a tester $T$ to be $f$-compatible with some
$f:\Sigma\to \Delta^k$.

\begin{prp}\label{PR:sep_nec_and_suff}
	Let $C\subseteq \Sigma^n$ be a code,
	let $T$ be a $q$-tester for $C$,
	and let $\Delta$ be another alphabet.
	\begin{enumerate}[label=(\roman*)]
		\item  If there is a function
		$f:\Sigma\to \Delta^k$ such that $T$ is $f$-compatible,
		then $T$ is $\Delta$-separable.
		\item If $T$ is $\Delta$-separable, then there
		is $k\in\N$ and an injective function $f:\Sigma\to \Delta^k$
		such that $T$ is $f$-compatible.
		In fact, $f$ can be constructed as follows: Let 
		$k = |\Delta^\Sigma|$, let $f_1,\dots,f_k$
		be all the functions from $\Sigma$ to $\Delta$,
		and take $f(w)=(f_1(w),\dots,f_k(w))$.
	\end{enumerate}
\end{prp}

\begin{proof}
	Write $T=(T^{(i)},(a^{(i)}_1,\dots,a^{(i)}_q),p_i)_{i\in I}$.

	(i) 
	For $j\in [k]$, let $f_j:\Sigma\to \Delta$
	denote the $j$-th component of $f$.
	Let $i\in I$.
	Our assumption that $T$ is $f$-compatible
	means that there exist a function  $g :\Delta^q\to \{0,1\}$
	and $b_1,\dots,b_q\in [k]$
	such that for every $w_1,\dots,w_q\in \Sigma$, 
	we have $T^{(i)}(w_1,\dots,w_q)=g (f_{b_1}(w_1),\dots,f_{b_q}(w_q))$.
	This is equivalent to saying that
	$T^{(i)}=g\circ (f_{b_1}\times \dots\times f_{b_q})$,
	so $T^{(i)}$ factors via $f_{b_1}\times \dots\times f_{b_q}:\Sigma^q\to\Delta^q$, and we have shown that $T$ is $\Delta$-separable.
	
	(ii) Define $f$ as in the statement.
	We need to show that $T$ is $f$-compatible.
	Let $i\in I$.
	Then there are functions $g_1,\dots,g_q:\Sigma\to \Delta$
	and a function $g:\Delta^q\to \{0,1\}$
	such that $g\circ (g_1\times \dots\times g_q)=T^{(i)}$.
	By the construction of $f$,
	there are $b_1,\dots,b_q\in [k]$
	such that $g_\ell= f_{b_\ell}$ for all $\ell\in [q]$.
	Then $g\circ (f_{b_1}\times \dots\times f_{b_q})=T^{(i)}$,
	or rather,
	$T^{(i)}(w_1,\dots,w_q)=g(f_{b_1}(w_1),\dots,f_{b_q}(w_q))$
	for all $w\in\Sigma^q$. This is exactly what we need to show.
\end{proof}

\begin{prp}\label{PR:sep_nec_and_suff_linear}
	Let $C\subseteq \Sigma^n$ be an $\F$-code,
	let $T$ be an $\F$-linear  $q$-tester for $T$,
	and let $\Delta$ be a  nonzero $\F$-vector space.
	\begin{enumerate}[label=(\roman*)]
		\item  If there is an $\F$-linear function
		$f:\Sigma\to \Delta^k$ such that $T$ is $f$-compatible,
		then $T$ is linearly $\Delta$-separable.
		\item If $T$ is linearly $\Delta$-separable, then there
		is $k\in\N$ and an injective function $f:\Sigma\to \Delta^k$
		such that $T$ is $f$-compatible.
		In fact, $f$ can be constructed as follows: Let 
		$k = |\Hom_{\F}(\Sigma,\Delta)|$, let $f_1,\dots,f_k$
		be all the linear functions from $\Sigma$ to $\Delta$,
		and take $f(w)=(f_1(w),\dots,f_k(w))$.
	\end{enumerate}
\end{prp}

\begin{proof}
	This is completely analogous to the proof of
	Proposition~\ref{PR:sep_nec_and_suff}.
\end{proof}

\begin{remark}\label{RM:image}
	In Proposition~\ref{PR:sep_nec_and_suff}(ii), the image
	of $f:\Sigma \to \Delta^k$ is the generalized  long code
	$L(\Sigma,\Delta)$ of \S\ref{subsec:gen_long_code}.
	In Proposition~\ref{PR:sep_nec_and_suff_linear}(ii), the image
	of $f:\Sigma \to \Delta^k$ is  
	the generalized Hadamard code $H(\Sigma,\Delta)$ 
	of \S\ref{subsec:gen_Hadamard}.
\end{remark}

Next, we show that every code admitting a $q$-tester
also admits a $\Delta$-separable
$q$-tester with soundness that is proportional to that
of the original tester.
In particular, when given an LTC, we can always assume that its
tester is $\Delta$-separable.

\begin{thm}\label{TH:reduction_to_sep}
	Let $C\subseteq \Sigma^n$ be a code,
	let $T$ be a possibly adaptive $q$-tester for $C$ with soundness $\mu>0$,
	and let $\Delta$ be another alphabet. 
	Then $C$ admits a $\Delta$-separable tester $T'$  with soundness $\frac{\mu}{|\Sigma|^q}$.
\end{thm}

This proof is a mild modification of a standard trick to turn $T$ into a non-adaptive
	$q$-tester. The idea is to guess   what will be  the $q$ letters
	that the adaptive tester $T$ would see. Based on this guess,
	we can   determine which
	coordinates $T$ would read, and then read letters in those positions.
	If the letters did match our guess, we return what $T$ would return, and
	if our guess was wrong, we accept the word.

\begin{proof}
	Let $I$ be the set of random seeds for $T$
	(say, $I=\{0,1\}^r$ if $T$ flips $r$ coins).
	We define a non-adaptive tester  $T'$ as follows:
	Given $w\in \Sigma^n$, choose   $u=(u_1,\dots,u_q) \in \Sigma^q$ 
	and $i\in I$ uniformly
	at random. 
	Let $a^{(i,u)}_1,\dots,a^{(i,u)}_q$ be the coordinates
	that $T$ would have read on the seed $i$
	if the letters it would query in these positions were $u_1,u_2,\dots$,
	and let $c$ be the output of $T$ in that situation.
	Now query $w_{a^{(i,u)}_1},\dots,w_{a^{(i,u)}_q}$.
	If this sequence is different from $u_1,\dots,u_q$,
	then $w$ is accepted. Otherwise, $T'$ returns $c$.

	It is straightforward
	to see that $\Prob(T'(w)=0)\geq \frac{1}{|\Sigma^q|}\Prob(T(w)=0)$,
	so $T'$
	has soundness $\frac{\mu}{|\Sigma|^q}$.
	It remains to show that $T'$
	is $\Delta$-separable.
	To that end,
	note that
	\[
	T' = \Circs{T^{(i,u)},(a^{(i,u)}_1,\dots,a^{(i,u)}_q),{\textstyle\frac{1}{|I|\cdot|\Sigma|^q}}}_{(i,u)\in I\times \Sigma^q},
	\]
	where $T^{(i,u)}:\Sigma^q\to \{0,1\}$ 
	a function that returns $1$ for every $u'\in \Sigma^q -\{u\}$
	and for $u$ returns   what $T$ would have returned
	on the random seed $i$ if it would read the letters
	$u_1,u_2,\dots$ in positions $a^{(i,u)}_1,a^{(i,u)}_2,\dots$.
	
	Fix $(i,u)\in I\times \Sigma^q$.
	We need to show that $T^{(i,u)}=g\circ (g_1\times\dots\times g_q)$
	for some $g_1,\dots,g_q:\Sigma\to \Delta$
	and $g:\Delta^q\to \{0,1\}$.
	Without loss of generality,
	we may assume that $0,1\in\Delta$.
	If $T^{(i,u)}(u)=1$,
	then $T^{(i,u)}$ is identically
	$1$ and we can take $g_1,\dots,g_q,g$
	to be the  functions which map everything to $1$.
	Otherwise,   $T^{(i,u)}(u)=0$ and
	$T^{(i,u)}(w)=1$ for all  $w\in \Sigma^q-\{u\}$.
	In this case, define $g_\ell:\Sigma\to \Delta$
	by
	\[
	g_\ell(w)=
	\left\{
	\begin{array}{ll}
	0 & w_\ell = u_\ell \\
	1 & w_\ell \neq u_\ell
	\end{array}
	\right.
	\]
	and let $g:\Delta^q\to\{0,1\}$
	be the function which maps $(0,\dots,0)$ to $0$
	and all other elements to $1$.
	It is straightforward
	to check that $T^{(i,u)}=g\circ (g_1\times\dots\times g_q)$,
	so $T'$ is indeed $\Delta$-separable.
\end{proof}

\begin{thm}\label{TH:reduction_to_sep_lin}
	Let $C\subseteq \Sigma^n$ be an $\F$-code,
	let $T$ be an $\F$-linear $q$-tester for $C$ with soundness $\mu>0$,
	and let $\Delta$ be a nonzero $\F$-vector space. 
	Then $C$ admits a linearly $\Delta$-separable tester $T'$  with soundness $ \Ceil{\frac{q \dim \Sigma}{\dim \Delta}}^{-1}\mu$.
\end{thm}


\begin{proof}
	Write $T=(T^{(i)},(a^{(i)}_1,\dots,a^{(i)}_q),p_i)_{i\in I}$
	and put $m=\ceil{\frac{q \dim\Sigma}{\dim \Delta}}$.
	Then for every $i\in I$, there is a subspace
	$V_i\subseteq \Sigma^q$
	such that $T^{(i)}(u)=1$ if and only if $u\in V_i$ ($u\in \Sigma^q$).
	By the definition of $m$, we have $\dim \Delta^m=m\dim \Delta\geq \dim \Sigma^q$.
	Thus, there is a linear map $h_i:\Sigma^q\to \Delta^m$ with
	kernel   $V_i$.
	For $j\in [m]$, denote by $h_{ij}:\Sigma^q\to \Delta$ the $j$-th component of $h_i$.
	
	We now define an $\F$-linear $q$-tester $T'$ as follows:
	Given $w\in \Sigma^n$, choose $i\in I$ according to $p$ and $j\in [m]$
	uniformly at random, and accept $w$ if and only if $h_{ij}(w_{a_1^{(i)}},\dots,w_{a_q^{(i)}})=0$.
	It is clear that $T'$ accepts every $w\in C$.
	Moreover, since $V_i=\ker h_i=\bigcap_{j=1}^m \ker h_{ij}$,
	whenever $T^{(i)}(w_{a_1^{(i)}},\dots,w_{a_q^{(i)}})=0$,  there
	is at least one $j\in[m]$ such that $h_{ij}(w_{a_1^{(i)}},\dots,w_{a_q^{(i)}})\neq 0$.
	This means that $\Prob(T'(w)=0)\geq \frac{1}{m}\Prob(T(w)=0)$,
	so $T'$ has soundness $\frac{\mu}{m}$.
	
	It remains to show that $T'$ is linearly $\Delta$-separable.
	To that end, observe that
	\[
	T'=(T^{(i,j)},(a^{(i,j)}_1,\dots,a^{(i,j)}_q),p_{(i,j)})_{(i,j)\in I\times [m]},
	\]
	where $p_{(i,j)}:=\frac{p_i}{m}$, $a^{(i,j)}_\ell :=a^{(i)}_\ell$
	($i\in I$, $j\in [m]$, $\ell\in[q]$),
	and $T^{(i,j)}(u)$ ($u\in\Sigma^q$)
	is defined to be $1$ if $u\in\ker h_{ij}$ and $0$ otherwise.
	Fix some $i\in I$ and $j\in [m]$.
	Then there are linear maps $g_1,\dots,g_q:\Sigma\to \Delta$
	such that $h_{ij}(u)=g_1(u_1)+\dots+g_q(u_q)$ for all $u\in\Sigma^q$.
	This means that we can recover $T^{(i,j)}(u)$ from
	$g_1(u_1),\dots,g_q(u_q)$, so $T^{(i,j)}$ factors
	via $g_1\times\dots\times g_q:\Sigma^q\to \Delta^q$.
\end{proof}

\section{Alphabet Reduction for Good  LTCs}
\label{sec:alpha_red}

We finally put the results of the previous sections together to prove
our alphabet reduction results for good LTCs.

We begin with results about $\F$-codes.
Theorem~\ref{TH:main_lin} is an immediate consequence of the following theorem
applied with $(q,c)=(2,2)$ and $(q,c)=(3,1)$.

\begin{thm}\label{TH:linear_alphabet_reduction}
	Let $q\geq 2$ be an integer and let $\Delta$
	be an $\F$-vector space of dimension $d>0$.
	We further fix an  auxiliary parameter $c\in\{1,\dots,d\}$ and require
	that $q\geq 3$ if $c=1$.
	Suppose that $C\subseteq \Sigma^n$ is an $\F$-code  with relative
	distance $\delta$ and rate $r$
	admitting an $\F$-linear $q$-query  tester
	$T$ with soundness $\mu>0$.
	Then there exists an $\F$-code $C'\subseteq \Delta^{|\Sigma|^cn}$
	with relative distance $(1-|\F|^{-c})\delta$
	and rate $\frac{1}{d}\cdot \frac{\dim\Sigma}{|\Sigma|^c}\cdot r$ admitting an $\F$-linear $q$-query tester
	$T'$ with soundness
	\[
	\frac{c\mu\nu}{(q\dim \Sigma+c)(q|\Sigma|^c+\nu)+c\mu(1+\nu)},	
	\]
	where
	$\nu$ is the largest possible soundness
	of a $q$-query tester for the
	generalized Hadamard code $H(\Sigma,\F^c)$ (see \S\ref{subsec:gen_Hadamard}).
	Moreover, $\nu\geq  |\Sigma|^{-2c}$ if $c>1$ and $\nu\geq |\Sigma|^{-3}$
	if $c=1$.
\end{thm}

Our requirement about $q$ implies that the theorem cannot be applied
with $1$-dimensional $\Delta$ and $q=2$.

Note that increasing the auxiliary parameter $c$ decreases the rate of
$C'$, and also the ratio between the block-lengths of $C$ and $C'$.
On the other hand, we expect that increasing $c$ would increase $\nu$, and thus the soundness
of $T'$.
In particular, we believe that $\nu$ is much larger than the naive lower bound provided in the theorem.
For example, when $q=3$ and $c=1$, we have $\nu\geq \frac{1}{6}$
by \cite[Lem.~B.2]{Meir_2009_comb_const_of_LTCs};
see also \cite[Thm.~4.8]{Mittal_2025_low_soundness_testing} which shows that $\nu$
is close to $1$ when $c=1$.

\begin{proof}
	Fix a $c$-dimensional subspace $\Delta'\subseteq \Delta$.
	By Theorem~\ref{TH:reduction_to_sep_lin}, we may assume that
	$T$ is $\Delta'$-separable
	at the cost of reducing the soundness from $\mu$ to $\mu':=\frac{c}{q\dim \Sigma+c}\mu$
	(note that $\ceil{\frac{q\dim \Sigma}{c}}^{-1}\geq
	(\frac{q\dim \Sigma}{c}+1)^{-1}=\frac{c}{q\dim \Sigma+c}$).
	Now, by Proposition~\ref{PR:sep_nec_and_suff_linear}(ii),
	there is an $\F$-linear injective function $f:\Sigma\to (\Delta')^k$
	such that $T$ is $f$-compatible. Moreover, by Remark~\ref{RM:image},
	$D:=\im(f)$ is the generalized Hadamard code $H(\Sigma,\Delta')$, and so $k=|\Sigma|^c$.
	By the definition of $\nu$, the code $D\subseteq (\Delta')^k$
	admits an $\F$-linear $q$-tester with soundness $\nu$.
	Moreover, since $q\geq 2$ when $c>1$ (resp.\ $q\geq 3$ with $c=1$),
	$\nu$ is greater than or equal
	to the soundness of the $2$-dependence (resp.\ $3$-dependence) tester
	for $H(\Sigma,\Delta')$. 
	The latter is positive by  Theorem~\ref{TH:gen_Hadamard_2_test}
	(resp.\ the discussion in \S\ref{subsec:gen_Hadamard}),
	so $\nu\geq k^{-2}=|\Sigma|^{-2c}$ (resp.\ $\nu\geq k^{-3}=|\Sigma|^{-3}$).
	
	Consider the code $C'':=C\circ_f D\subseteq (\Delta')^{kn}$.
	By Proposition~\ref{PR:concat_rate_dist},
	it has relative distance $\delta(H(\Sigma,\Delta'))\delta(C)=(1-|\Delta'|^{-1})\delta$ and
	rate $r(H(\Sigma,\Delta')) r(C)=\frac{r\dim\Sigma}{c|\Sigma|^c}$, and by Theorem~\ref{TH:comp_of_codes_w_testers},
	it has an $\F$-linear $q$-tester with soundness
	\[
	\mu'':=\frac{\mu'\nu}{ qk+\mu'+\nu}=\frac{c\mu\nu}{(q\dim \Sigma+c)(q|\Sigma|^c+\nu)+c\mu}.
	\] 
	Finally, we take $C'$ to be $C''$ viewed as a code
	with alphabet $\Delta$ (instead of $\Delta'$).
	This multiplies the rate
	by $\frac{\dim\Delta'}{\dim \Delta}=\frac{c}{d}$,
	has no effect on the relative distance,
	and   by Proposition~\ref{PR:alphabet_increase}, 
	$C'$ still admits an  $\F$-linear $q$-tester with soundness
	$\frac{\mu''}{\mu''+1}=\frac{c\mu\nu}{(q\dim \Sigma+c)(q|\Sigma|^c+\nu)+c\mu(1+\nu)}$.
\end{proof}

Theorem~\ref{TH:main_non_lin} follows by applying following theorem
with $(q,c)=(2,3)$ and $(q,c)=(3,2)$.

\begin{thm}\label{TH:alphabet_red}
	Let $q\geq 2$ be an integer and let $\Delta$
	be an alphabet with $d $ elements.
	We further fix an  auxiliary parameter $c\in\{2,\dots,d\}$ and require
	that $q\geq 3$ if $c=2$.
	Suppose that $C\subseteq \Sigma^n$ is a code  with relative
	distance $\delta$ and rate $r$
	admitting a possibly adaptive $q$-query  tester
	$T$ with soundness $\mu>0$.
	Then there exists a code $C'\subseteq \Delta^{c^{|\Sigma|}n}$
	with relative distance $(1-c^{-1})\delta$
	and rate $\frac{\log_d|\Sigma|}{c^{|\Sigma|}}\cdot r$ admitting a $q$-query tester
	$T'$ with soundness
	\[
	\frac{\mu\nu}{|\Sigma|^q(qc^{|\Sigma|}+\nu)+\mu(\nu+1)}	,
	\]
	where
	$\nu$ is the largest possible soundness
	of a $q$-query tester for the
	generalized long code $L(\Sigma,\{1,\dots,c\})$ (see \S\ref{subsec:gen_long_code}).
	Moreover, $\nu\geq c^{-2|\Sigma|}$ if $c>2$ and $\nu\geq 2^{-3|\Sigma|}$
	if $c=2$. 
\end{thm}

Our requirement about $q$ implies that the theorem cannot be applied
when $q=2$ and $|\Delta|=2$.
As in Theorem~\ref{TH:linear_alphabet_reduction},
the parameter $c$ allows a trade-off between code-theoretic properties of $C'$ 
and the   soundness of the tester $T'$,
where
again, we believe that $\nu$ is much larger than the naive bounds given in the theorem.

\begin{proof}
	The proof strategy is similar to that of proving Theorem~\ref{TH:linear_alphabet_reduction}.
	
	Let $\Delta'$ be a subset of $\Delta$ having $c$ elements.
	We use Theorem~\ref{TH:reduction_to_sep}
	to replace $T$ with a $\Delta'$-separable tester with soundness
	$\mu':=\frac{\mu}{|\Sigma|^q}$.
	By Proposition~\ref{PR:sep_nec_and_suff}(ii),
	there is an   injective function $f:\Sigma\to (\Delta')^k$
	such that $T$ is $f$-compatible, and by Remark~\ref{RM:image},
	it can be chosen so that $D:=\im(f)$ is the generalized long code $L(\Sigma,\Delta')$.
	In particular,  $k=c^{|\Sigma|}$.
	By Theorem~\ref{TH:gen_long_code_2_test} 
	(resp.\ Proposition~\ref{PR:long_code_3_test}),
	when $c>2$ (resp.\ $c=2$), the $2$-dependence (resp.\ $3$-dependence) for $D$
	has positive soundness.
	As a result, when $c>2$ (resp.\ $c=2$), $D$ has a $q$-tester
	with soundness $\nu\geq k^{-2}$ (resp.\ $\nu\geq k^{-3}$).

	The proof is finished as in Theorem~\ref{TH:linear_alphabet_reduction}
	by taking $C'$ to be the code $C\circ_f D\subseteq (\Delta')^{kn}$
	and viewing it as a code inside $\Delta^{kn}$.
\end{proof}

Next, we note that the parameters of the code $C'$ and the tester $T'$
promised by  Theorem~\ref{TH:alphabet_red}
can be improved if $C$ is assumed to be an $\F_2$-code and $T$ is $\F_2$-linear.
This is the content of the following theorem, which we state only in the case
$|\Delta|=3$ for simplicity.

\begin{thm}\label{TH:alphabet_red_semilin}
	Suppose that $C\subseteq \Sigma^n$ is an $\F_2$-code  with relative
	distance $\delta$ and rate $r$
	admitting an $\F_2$-linear $q$-query  tester
	$T$ with soundness $\mu>0$.
	Then there exists a code $C'\subseteq \{0,1,2\}^{(2|\Sigma|^2+|\Sigma|)n}$
	with relative distance $\frac{\delta}{2|\Sigma|^2+|\Sigma|}$
	and rate $\frac{\log_3|\Sigma|}{2|\Sigma|^2+|\Sigma|}\cdot r$ admitting a $q$-query tester
	$T'$ with soundness
	\[
	\frac{\mu\nu}{q\dim \Sigma \cdot(2q|\Sigma|^2+q|\Sigma|+\nu)+\mu}
	\]
	where
	$\nu\geq (2|\Sigma|^2+|\Sigma|)^{-2}$ is the maximal possible soundness 
	of a $2$-query tester for a certain code $D\subseteq\{0,1,2\}^{2|\Sigma|^2+|\Sigma|}$ specified in the proof. 
\end{thm}

We will need the following easy lemma.

\begin{lem}\label{LM:compatibility_extension}
	Let $C\subseteq \Sigma^n$
	be a code with a $q$-tester $T$.
	Let $\Delta$ be an alphabet and let $\Delta'$ be another alphabet
	containing $\Delta$.
	Let $f_1,\dots,f_k:\Sigma\to \Delta $
	and let $f'_1,\dots,f'_{k'}:\Sigma\to \Delta' $.
	Define $f:\Sigma\to \Delta^k$ by $f(a)=(f_1(a),\dots,f_k(a))$
	and define  $f':\Sigma\to (\Delta')^{k'}$ similarly.
	Suppose that for every $i\in [k]$, there is $i'\in [k']$
	such that $f'_{i'}=f_i$ as functions from $\Sigma$ to $\Delta'$.
	If $T$ is $f$-compatible, then it is $f'$-compatible.
\end{lem}

\begin{proof}
	Write $T=(T^{(i)},(a_1^{(i)},\dots,a_q^{(i)}),p_i)_{i\in I}$.
	Let $i\in I$.
	Then there are $b_1,\dots,b_q\in [k]$ and $g:\Delta^q\to \{0,1\}$
	such that $T^{(i)}(w)=g(f_{b_1}(w_1),\dots, f_{b_q}(w_q))$ for all $w\in \Sigma^q$.
	By assumption, there are $b'_1,\dots,b'_q\in [k']$
	such that $f'_{b'_j}=f_{b_j}$ for all $j\in [q]$.
	Now, if we extend $g$ in some way to a function $g':\Delta'\to \{0,1\}$,
	we have $T^{(i)}(w)=g'(f'_{b'_1}(w_1),\dots,f'_{b'_q}(w_q))$ for all $w\in \Sigma^q$.
	This means that $T$ is $f'$-compatible.
\end{proof}

\begin{proof}[Proof of Theorem~\ref{TH:alphabet_red_semilin}]
	By Theorem~\ref{TH:reduction_to_sep_lin},
	we may replace  $T$ with a linearly $\F_2$-separable
	tester at the cost of reducing the soundness to 
	$\mu':=\frac{\mu}{q\dim\Sigma}$.
	
	Let $g_1,\dots,g_t$
	be all the linear $\F_2$-functions
	from $\Sigma$ to $\F_2$ (so $t=|\Sigma|$),
	and let $g:\Sigma\to \{0,1\}^t$ be defined by
	$g(a)=(g_1(a),\dots,g_t(a))$.
	By Proposition~\ref{PR:sep_nec_and_suff_linear}(ii),
	$T$ is $g$-compatible.
	Now define $f_1,\dots,f_\ell:\Sigma\to \{0,1,2\}$ as in Lemma~\ref{LM:critical_lemma} (with $\ell$ in place of $n$). By construction,
	$\ell\leq t+2t^2=:k$, and if $\ell<t+2t^2$ we define
	$f_{\ell+1},\dots,f_k:\Sigma\to\{0,1,2\}$ to be identically $0$.
	Now define $f:\Sigma\to\{0,1,2\}^k$ similarly to $g$.
	By Lemma~\ref{LM:compatibility_extension}, $T$
	is also $f$-compatible.
	
	Put $D=D(\{f_i\}_{i=1}^k)\subseteq \{0,1,2\}^k$ and let $C'=C\circ_f D\subseteq \{0,1,2\}^{kn}$.
	By Lemma~\ref{LM:critical_lemma}, $D\subseteq \{0,1,2\}^k$
	has a $2$-query tester with soundness $ k^{-2}=(t+2t^2)^{-2}$, so $\nu\geq (t+2t^2)^{-2}$.
	Furthermore, by Theorem~\ref{TH:comp_of_codes_w_testers},
	$C'\subseteq \{0,1,2\}^{kn}$ has a $2$-query tester with soundness
	\[
	\frac{\mu'\nu}{qk+\mu'+\nu}=\frac{\mu\nu}{q\dim \Sigma \cdot(2q|\Sigma|^2+q|\Sigma|+\nu)+\mu}.
	\]
	To finish, note that by Proposition~\ref{PR:concat_rate_dist},
	$\delta(C')\geq \delta(D)\delta(C)\geq \frac{\delta}{k}$
	and $r(C')= r(D)r(C)=\frac{\log_3| \Sigma|}{k}\cdot r$.
\end{proof}

\begin{remark}\label{RM:randomness_comp}
	By tracing the proofs of Theorems~\ref{TH:linear_alphabet_reduction}, \ref{TH:alphabet_red},
	\ref{TH:alphabet_red_semilin},
	one sees that if the tester $T$ of $C$ has randomness complexity $R$,
	then the promised tester $T'$ for $C'$ has  randomness complexity $\max\{R,\log_2 n\}+O(1)$,
	where the $O(1)$-factor depends on $\mu$, $|\Sigma|$, $|\Delta|$.
\end{remark}

\bibliographystyle{alpha}
\bibliography{MyBib_24_03}

\end{document}